\documentclass[12pt, draftclsnofoot, onecolumn]{IEEEtran}

\usepackage{amsfonts}
\usepackage{dsfont}
\usepackage{setspace}
\usepackage{color}
\usepackage{amssymb}
\usepackage{cite}
\usepackage[cmex10]{amsmath}
\usepackage{algorithm, algpseudocode}
\usepackage{algorithmicx}
\usepackage{array}
\usepackage{mathrsfs}
\usepackage{graphicx}
\usepackage{graphics}
\usepackage{latexsym}
\usepackage{amscd}
\usepackage{amsfonts}
\usepackage{amsmath,amscd,amssymb,verbatim}
\usepackage{graphics}
\usepackage{amsthm}
\usepackage[T1]{fontenc}
\usepackage[utf8]{inputenc}
\usepackage{authblk}
\usepackage[short]{optidef}
\usepackage{mathtools}
\usepackage{soul}
\usepackage{anyfontsize}
\usepackage{textcomp}
\usepackage{xcolor}
\usepackage{bbm}
\usepackage{aligned-overset}
\usepackage{enumerate}

\usepackage{blindtext}
\IEEEoverridecommandlockouts

\newcommand{\Cren}{\mathcal{C}_{2|1,n}}
\newcommand{\Vren}{V_{2|1,n}}
\newcommand{\wren}{\omega_{2|1,n}}

\newtheorem{theorem}{Lemma}

\begin{document}
\renewcommand{\thefootnote}{\fnsymbol{footnote}}
\title{Low-Latency Hybrid NOMA-TDMA: 
QoS-Driven Design Framework
\thanks{Y. Zhu, Y. Hu, Xiaopeng Yuan and A. Schmeink are with the  ISEK Research Area/Lab, RWTH Aachen University, 52074 Aachen, Germany (e-mail: \{zhu,hu,yuan,schmeink\}@isek.rwth-aachen.de). T. Wang is with the School of Electronics and Information Engineering, Harbin Institute of Technology, Shenzhen 518055, China (e-mail: tongwang@hit.edu.cn). M. C. Gursoy is with the Department of Electrical Engineering and
Computer Science, Syracuse University, NY 13210, USA (e-mail: mcgursoy@syr.edu). Y. Hu is the corresponding  author. } }
\author{
Yao Zhu, %
Xiaopeng Yuan, %
Yulin Hu, 
Tong Wang, \\\vspace{-15pt} %
M. Cenk Gursoy, %
Anke Schmeink%
\vspace*{-1.65cm} 
}
\maketitle
\begin{abstract}
Enabling ultra-reliable and low-latency communication services while providing massive connectivity is one of the major goals to be accomplished in future wireless communication networks. %
In this paper, we investigate the performance of a hybrid multi-access scheme in the finite blocklength (FBL) regime that combines the advantages of both non-orthogonal multiple access (NOMA) and time-division multiple access (TDMA) schemes. Two latency-sensitive application scenarios  are studied, distinguished by whether the queuing behaviour has an influence on the transmission performance or not. In particular, for the latency-critical case with one-shot transmission, we aim at a certain physical-layer quality-of-service (QoS) performance, namely the optimization of the   reliability. And for the case in which queuing behaviour plays a role, we focus on the link-layer QoS performance and provide a design that maximizes the effective capacity.
For both designs, we leverage the characterizations in the FBL regime to provide the optimal framework by jointly allocating the blocklength and transmit power of each user. In particular, for the reliability-oriented design, the original problem is decomposed and the joint convexity of sub-problems is shown via a variable substitution method. For the effective-capacity-oriented design, we exploit the method of Lagrange multipliers to formulate a solvable dual problem with strong duality to the original problem. 
Via simulations, we validate our analytical results of convexity/concavity and show the advantage of our proposed approaches compared to other existing schemes.
\end{abstract}
\vspace*{-0.3cm}
\vspace{-10pt}
\section{Introduction}
In future 6G wireless communication networks, the  ultra-reliable and low-latency  communication (URLLC) 
services are expected to be integrated with massive connectivity.
{This is fundamentally important, since it will enable a wide range of novel delay-sensitive and mission-critical applications, e.g., autonomous  driving, augmented/virtual/mixed reality   and factory automation, by satisfying massive number of users' stringent requirements on the latency and error probability~\cite{intro_mMTC, Meng2022Sampling}.}

On the one hand, due to the stringent latency requirements of URLLC,   
  short packet communications are likely to be employed, where the transmissions are carried out via so-called finite blocklength (FBL) codes~\cite{URLLC_intro}.
  Under this FBL assumption, data transmissions are no longer arbitrarily reliable, especially when the blocklength is short or  signal-to-noise ratio (SNR) is low. 
  In the landmark work of Polyanskiy \emph{et al.}~\cite{Polyanskiy_2010}, the maximal achievable coding rate as a function of target error probability is expressed in closed-form in AWGN channels. 
Subsequently, the FBL performance analysis has been   extended to Gilbert-Elliott channels~\cite{Polyanskiy_2011}, flat-fading channels~\cite{Yang_2014} and random access channels~\cite{FBL_random_channel}. Based on those models, FBL performance characterizations have also been widely addressed in various wireless single-user networks including MIMO~\cite{FBL_MIMO}, relaying~\cite{Hu_TWC_2015}, energy-harvesting~\cite{FBL_EH1}
and {quality-of-service (QoS) constrained downlink networks~\cite{FBL_QoS1}}.

On the other hand,  {enhancing massive connectivity  is another major concern  in 6G~\cite{intro_next}.} 
{In fact, in comparison to 5G, one of the key challenging tasks  in 6G is to provide URLLC to massive users with limited resources. This new service class is also referred to as mURLLC~\cite{Chen_MA_mURLLC,Zhang_mURLLC_FBL}. In order to overcome this challenge, many novel multiple access schemes are considered,} where 
time-division multiple access (TDMA) is one of the most popular orthogonal multiple access (OMA) schemes that is able to efficiently provide wireless access to multiple users via dynamically allocating the length of each slot~\cite{TDMA}.
However, TDMA alone can be inefficient to address massive connectivity in the FBL regime, since it divides the entire frame (which is already short due to the low latency requirements) into slots with even shorter blocklength, which can  deteriorate the performance. {According to~\cite{Polyanskiy_2010}, extremely short blocklength leads to poor reliability even if the channel link has a high  quality~\cite{FBL_access}.}
This motivated us to investigate non-orthogonal multiple access (NOMA) techniques, which are envisioned as a potential solution and has been widely studied in both industry and academia~\cite{intro_NOMA_3GPP,intro_NOMA_1,intro_NOMA_2,intro_NOMA_3}. 
The key advantage of NOMA compared with conventional OMA schemes is the high spectrum efficiency achieved by sharing the bandwidth among users (and hence not requiring to split the blocklength). 
Recently, the FBL performance of NOMA in different system setups has been investigated~\cite{Sun_joint_NOMA,Hu_coop_NOMA,Lai_coop_NOMA, Han_massive_NOMA}. 
In particular, the authors in~\cite{Sun_joint_NOMA} provide an optimal joint 
design of power control and rate adaptation, where the throughput of one NOMA user is maximized while guaranteeing the throughput constraint of the other NOMA user. 
In~\cite{Hu_coop_NOMA}, a cooperative NOMA scheme in the FBL regime is studied, where the NOMA performance is enhanced via relaying. 
The authors in~\cite{Lai_coop_NOMA} further analyze the average error probability with flat Rayleigh fading channels via a linear approximation of Polyanskiy's FBL model. 
Moreover, a framework with multiple NOMA power levels is proposed in~\cite{Han_massive_NOMA}, aiming at supporting massive connectivity, where a joint blocklength and power allocation is optimized via a reinforcement learning approach.

However, NOMA scheme also introduces interference among the users and relies on the performance of successive interference cancellation (SIC)~\cite{intro_NOMA_SIC}. This will significantly influence the QoS of the transmissions, e.g., the reliability and the delay due to retransmissions.  %
{Based on above observations, it is essential to combine the advantage of both NOMA and  TDMA to strike a balance between longer blocklength and lower/no interference to enhance the massive access while providing URLLC services.} {The so-called hybrid NOMA-TDMA scheme is proposed initially in a mobile edge computing scenario, which is a combination of the uplink NOMA scheme within the user pairs and the TDMA scheme between the user pairs~\cite{hybrid_2}.}
Subsequently, it is also applied in energy harvesting-enabled systems~\cite{hybrid_1}, machine-to-machine communications~\cite{hybrid_3}, as well as UAV-assisted networks~\cite{hybrid_4}. 
{ The major advantage of the hybrid NOMA-TDMA compared with existing massive access scheme is the increased access capacity with limited resources while fulfilling the heterogeneous QoS requirements~\cite{intro_MA,intro_She}. More importantly, it is realized without additional energy cost since the IoT devices have generally low power.  However, how to provide an optimal design for the hybrid NOMA-TDMA scheme while taking into account the impact of operating in the FBL regime is still an open challenge in massive access, which is not yet addressed in the literature.} 
{For example, the impact of imperfect SIC due to the decoding error, where the error probability will propagate from the higher power level user to lower power level user, is often underestimated or even overlooked. Moreover, the blocklength allocation for each time slot in the TDMA scheme in the FBL regime is also non-trivial due to the non-linear correlation between blocklength and reliability. In particular, %
{how to address the aforementioned trade-off between TDMA slot length and NOMA interference by optimally allocating the transmit power and blocklength of each user with limited radio resources is still an open problem. This issue is especially critical in the multi-user scenarios, where the scalability of the allocation schemes should also be considered. }
More interestingly, the interplay of blocklength and transmit power allocation for the hybrid NOMA-TDMA scheme heavily influences the overall system performance, which should be carefully characterized. }

In this paper, we provide a design framework for the hybrid NOMA-TDMA scheme in a  multi-user uplink network operating with  FBL codes. {In particular, two typical scenarios based on different queue-behaviors are considered: $i)$ First, we address the queue-free scenario in which  mission-critical information with limited data size is generated and transmitted within a frame,} and our  design aims at a reliability-oriented physical-layer performance, specifically, the minimization of the maximal error probability among users;
{$ii)$ In the queue-influenced scenario, data arrives continuously at each user and is transmitted under given QoS requirements, in which the queuing behavior plays a role.}  In such a case, our    design focuses on a link-layer performance, i.e.,    the maximization of the effective capacity.
Our main contributions are summarized as follows:
\begin{itemize}
    \item For reliability-oriented design, we formulate an optimization problem which jointly allocates the blocklength of each slot and the transmit power of users in each pair, with the objective to minimize the maximal error probability among users. Since the optimization problem is non-convex, we transfer it from time domain into channel state domain. We further decompose the transferred problem. 
    {Moreover, for the first time, we rigorously prove the joint convexity of error probability (with respect to the power  of each NOMA user  and blocklength) based on the NOMA-influenced signal-to-interference-plus-noise ratio (SINR)} so that the obtained subproblems can be solved via a variable substitution method.
    \item For effective capacity-oriented design, we also formulate the optimization problem to maximize the sum of effective capacity of all users. However, instead of directly decomposing the problem, we  leverage the method of Lagrange multiplier to formulate a dual problem with strong duality. Then, we decompose the dual problem into subproblems.  We for the first time rigorously prove that the coding rate and effective capacity of the considered NOMA-TDMA scheme are jointly concave with respect to SINR and blocklength. %
    Based on our analytical findings, the subproblems can be characterized as convex problem with reformulated energy budget constraints.  As a result, the original problem can be solved iteratively.
    \item Via simulations, we show the significant advantages of our proposed framework compared to other benchmarks. The impact of various parameters is also discussed to provide guidelines for practical use cases.
\end{itemize}

The rest of the paper is organized as follows.  
Section II presents the  model of the considered hybrid NOMA-TDMA system. 
{Section~III provides the reliability-oriented design.} %
In Section IV, we address the effective capacity-oriented design and propose an approach to obtain the optimal solutions.
We provide numerical results in Section V and conclude the paper in Section~VI. 
\vspace{-10pt}
\section{System model}
\begin{figure}[!t]
\begin{minipage}[t]{0.47\linewidth}
    \centering
\includegraphics[width=0.5\textwidth,trim=20 10 25 0]{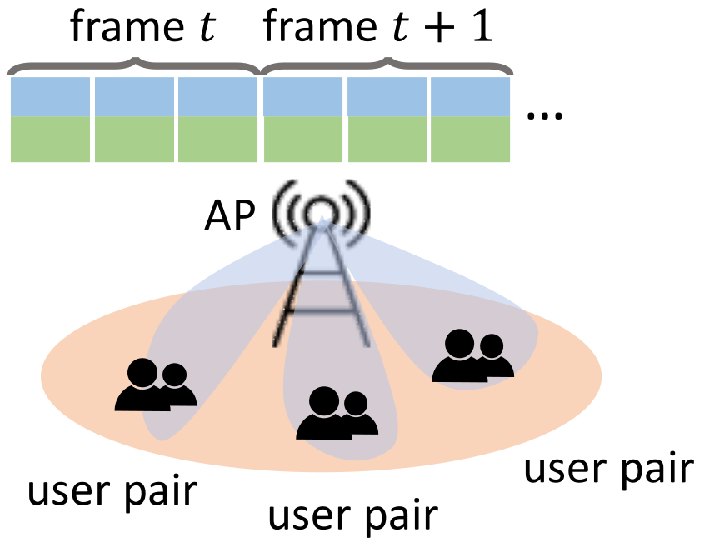}
\caption{ {An example of considered low-latency IoT networks. The transmissions are carried out frame-wise.}}
\label{fig:model}
\end{minipage}
\begin{minipage}[t]{0.005\linewidth}
    ~
\end{minipage}
\begin{minipage}[t]{0.47\linewidth}
    \centering
\includegraphics[width=0.8\textwidth,trim=20 10 25 15]{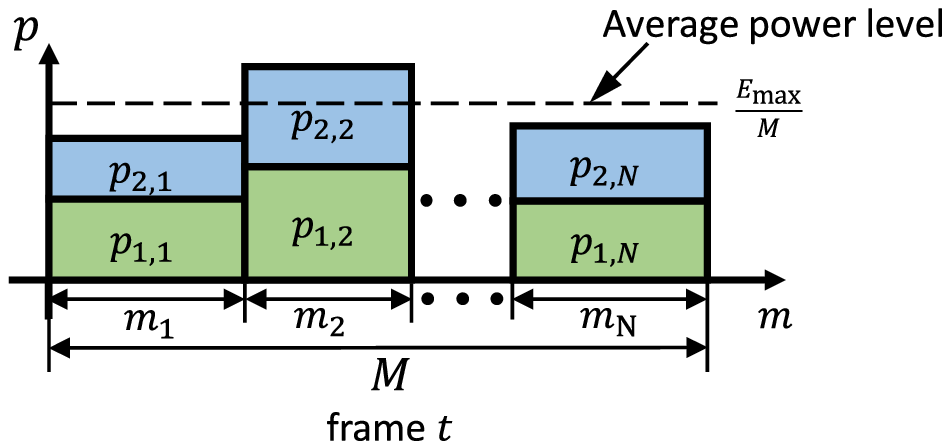}
\caption{ {Framework structure of considered hybrid NOMA-TDMA scheme in the time frame $t$. The available resource is restricted by both total blocklength $M$ and energy budget $E_{\max}$.}}
\label{fig:framework}
\end{minipage}
\vspace{-20pt}
\end{figure}
We consider a low latency IoT network with a server responsible for receiving delay-sensitive information from 
total $N$  user pairs as shown in Fig~\ref{fig:model}, where $n\in \mathcal{N}$ is the index of a user pair and $\mathcal{N}=\{1,\dots,N\}$ denotes its corresponding set.
Packets from users with the same size of $d$ bits are required to be transmitted to the server 
in each frame $t\in\mathcal{T}$ with a total available blocklength of $M$ (in symbols), where $\mathcal{T}=\{1,\dots,T\}$ is the set of time frame indices. The transmissions are carried out via a hybrid NOMA-TDMA scheme. 
{In particular, the transmissions between the server and different user pairs are operated in a TDMA manner\footnote{
{{In this work, we consider the communications are carried out with single carrier. However, it can also be carried out in a OFDMA manner, where we assign the subcarriers to each user pair instead of time slots. In fact, we can generalize the hybrid NOMA-TDMA scheme into a hybrid NOMA-OMA scheme by normalizing both the resources in frequency domain and in time domain. In this work, we focus on the hybrid NOMA-TDMA scheme, where the analysis can be applied to the hybrid NOMA-OFDMA scheme if we assume the channels are frequency-flat.}} }.  
{Therefore, the entire time duration is divided into $N$ time slots and each user pair is associated with one slot.} Let us denote by $\mathbf{m}=\{m_1,\dots,m_N\}$ the blocklengths assigned to each pair. %
Then,    $\sum^N_{i=1} m_i\leq M$ holds in order to guarantee the delay requirement. 
On the other hand, the transmissions between the server and two NOMA users in a given user pair share the same radio frequency band, i.e., uplink NOMA scheme is adopted for these transmissions.}
The channel information is assumed to be available at server side. 
{Hence, for each user pair, the server always regards the user with the lower channel gain as user 1 and the other is referred to as user 2,} i.e., $z_{1,n}\leq z_{2,n}$, where $z_{1,n}$ and~$z_{2,n}$ respectively denote the channel power gains (including path-loss) of user 1 %
and user 2 in the $n^{\rm th}$ pair. In this work, the channels of different users are assumed to be independent and
experience quasi-static fading, i.e., the channel state of each
link is constant during one block, and varies independently to 
the next. {Therefore, the order of users may also vary.}
{User 1 transmits the packets with power of $p_{1,n}$ and user 2 with $p_{2,n}$, while the maximal available transmit power for those users is $P_{\max}$.} %
Then, in each uplink NOMA transmission of the $n^{\rm th}$ user pair, the received signal at the server is given by 
\vspace{-.2cm}
\begin{equation}\vspace{-.2cm}
    \label{eq:signal_server}
    y_n=\sqrt{z_{1,n}p_{1,n}}x_{1,n}+\sqrt{z_{2,n}p_{2,n}}x_{2,n}+w_{n},
\end{equation}
where $x_{1,n}$ and $x_{2,n}$ are the transmitted information signals from user 1 and user 2, meanwhile~$w_{n}$ is the corresponding additive white Gaussian noise (AWGN) with zero mean and variance~$\sigma_{n}^2$. %

After receiving  signal $y_{n}$, the server first attempts to decode the signal for the stronger user  $x_{2,n}$ based on the SINR given by
\vspace{-.2cm}
\begin{equation}\vspace{-.2cm}
    \label{eq:snr_app}
    \gamma_{2|1,n}=\frac{z_{2,n}p_{2,n}}{z_{1,n}p_{1,n}+\sigma^2_{n}}\approx \frac{z_{2,n}p_{2,n}}{z_{1,n}p_{1,n}}.
\end{equation}
The approximation   follows the assumption  that interference is significantly stronger than the noise. {The validity and accuracy of this approximation will be shown in Section III and Section V.} 
After decoding the signal of user 2, server employs SIC to remove signal $x_{2,n}$ from $y_{n}$ and decodes signal from user 1 with  an SNR given by
    $\gamma_{1|1,n}=\frac{z_{1,n}p_{1,n}}{\sigma^2_{n}}.$
Note that  transmissions are carried out via FBL codes and are possible to be erroneous. 
{In such a case, user 1 has to decode its own signal directly based on the SINR of $\gamma_{1|2,n}=\frac{z_{1,n}p_{1,n}}{z_{2,n}p_{2,n}+\sigma^2_{n}}.$}

 \renewcommand*{\thefootnote}{\arabic{footnote}} 
\vspace{-.4cm}
\subsection{Transmission Rate with FBL Codes}
Due to the low-latency requirements, blocklength $m$ can no longer be regarded as infinite, precluding the direct use of Shannon capacity limit. To characterize the FBL performance more accurately, the authors in~\cite{Polyanskiy_2010} have derived the following tight bound on the maximal achievable transmission rate with target error probability $\bar\varepsilon$ %
in AWGN channels:  %
\vspace{-.2cm}
\begin{equation}\vspace{-.2cm}
\label{eq:maximal_rate}
    r^*\approx \mathcal{C(\gamma)}-\sqrt{\frac{V(\gamma)}{m}}Q^{-1}(\bar\varepsilon)\vspace{-5pt}
\end{equation}
 where  ${\mathcal{C}}(\gamma) = {\log _2}( {1 + \gamma } )$ is the Shannon capacity and ~$V(\gamma)$ is the channel dispersion~\cite{Chem_2015}. In the complex AWGN channel\footnote{{It is worthwhile to mention that the expression of $V$ may vary with the type of multi-access scheme of channel and depends on the coding scheme~\cite{FBL_random_channel}. For example, with i.i.d. Gaussian codes, channel dispersion is given by $V_{\mathrm{iid}}=\frac{\gamma}{1+\gamma}$~\cite{dispersion_V},  In this work, we adopt the widely used channel dispersion in~\cite{Ding_NOMA,Han_massive_NOMA}. However, it should be pointed out that our analytical results in the rest of section still hold with $V_{\mathrm{iid}}$.}}, $V(\gamma)=1- {(1+\gamma)^{-2}}$.  Moreover, $Q^{-1}(x)$ is the inverse Q-function with Q-function defined as $Q(x)=\int^\infty_x \frac{1}{\sqrt{2\pi}}e^{-\frac{t^2}{2}}dt$. 
Additionally, for any given data size $d$, according to~\eqref{eq:maximal_rate}, 
the (block) error probability %
for a single transmission is given by:
\vspace{-.2cm}
\begin{equation}\vspace{-.2cm}
\label{eq:error_tau}
{\textstyle
\varepsilon \!=\! {\mathcal{P}}(\gamma,\frac{d}{m},m)\! \approx\! Q\Big( {\sqrt {\frac{m}{V(\gamma)}} ( {{\mathcal{C}}(\gamma ) \!-\! \frac{d}{m})} \ln2} \Big)\mathrm{.}}\vspace{-5pt}
\end{equation}

{We consider such an error probability as the performance metric of the queue-free scenario.} 
\subsection{Effective Capacity}

{Let us consider users operating under constraints on the queueing delay in the queue-limited scenario. Therefore, a metric that takes queueing delay into account should be considered.} The transmissions of these users are subject to QoS constraints specified by a QoS exponent $\theta$. More specifically, let $Q$ denote the stationary queue length. Then, if we denote the queue threshold by $q$, the QoS exponent $\theta$ is defined as the decay rate of the tail of the distribution of queue length $Q$~\cite{EC_wu}:
\vspace{-.2cm}
\begin{equation}\vspace{-.2cm}
    \label{eq:theta}
    \theta=-\lim_{q\to\infty}\frac{\log \mathrm{Pr}(Q\geq q)}{q}.
\end{equation}

For a sufficiently large threshold $q_{\max}$, the buffer violation  probability can be approximated as:
\vspace{-.2cm}
\begin{equation}\vspace{-.2cm}
    \label{eq:delay_violation}
    \mathrm{Pr}(Q\geq q_{\max})\approx e^{-\theta q_{\max}}.
\end{equation}
Hence, large $\theta$ indicates a relatively strict QoS constraint, while small $\theta$ implies a loose one. %

Assume that the transmission system follows the general queuing model in~\cite{EC_wu}. Then, the average arrival rate in the queue must be equal to the average departure rate if the queue is in steady state. Let us denote the instantaneous arrival and service rates at the queue by $a$ and $s$, respectively. Then, in order for the buffer overflow probability to decay with rate  $\theta$ (or equivalently in order for (6) to hold), we have to satisfy the following condition:
\vspace{-.2cm}
\begin{equation}\vspace{-.2cm}
    \Lambda_a(\theta)+\Lambda_s(-\theta)=0.
\end{equation}
{Specially, for any random process $x\in\{a,s\}$,  $\Lambda_x(\theta)=\lim_{T\to\infty}\frac{1}{T}\log \mathbb{E}[e^{\theta X[T]}]$ with $X[T]=\sum^T_{t=1} x(t)$ is the asymptotic logarithmic moment
generating function (LMGF), where $\{x(t)|t=1,2,\dots\}$ denotes the discrete-time stationary and ergodic stochastic service progress~\cite{stability}.} Based on the LMGFs, the effective capacity,  which quantifies the maximum constant arrival rate that can be supported subject to the queuing constraint in~\eqref{eq:theta}, is given by:
\vspace{-.2cm}
\begin{equation}\vspace{-.2cm}
    \label{eq:EC_original}
    R=-\frac{\Lambda_s(-\theta)}{\theta}=-\lim_{T\to\infty}\frac{1}{\theta T}\log \mathbb{E}[e^{-\theta S[T]}].
\end{equation}

Considering the aforementioned system models and performance metrics, we discuss two different scenarios with practical aims and their corresponding design frameworks in the subsequent sections.
\vspace{-10pt}
\section{Reliability-Oriented Design Framework}
\label{sec:err}

{For the mission-critical applications, there are generally limited but latency-critical information bits for each user, e.g., state update from data sensing,  to be transmitted in the given slot. %
Usually, these small data packets have the highest priority, i.e., should be immediately transmitted without waiting in the queue.} 
And the  critical latency requirement does not allow retransmissions even when the transmissions fail.
{Hence, for such a queue-free scenario, the random impact of the queuing delay is negligible, and transmission delay is dominant.} 
{In other words, we can guarantee the critical delay requirements by appropriately choosing the blocklength.} 
Hence,  the other physical-layer QoS performance metric, i.e., the reliability, of such one-shot transmission becomes the key concern of the corresponding system design. 

In this section, we address the above issue  and provide a reliability-oriented framework design for such a scenario. 
In particular, we first characterize the error probability   for each user by taking into account the impact of both FBL and NOMA. 
Next, we formulate the optimization problem and transform it from the time domain into the channel state domain. Then, we decompose the transferred problem into several solvable sub-problems. 
{Finally, we provide the optimal solution   by investigating the joint convexity of error probability and characterizing the convexity of sub-problems with variable substitutions.} %
\subsection{Error Probability Characterization  and Problem Formulation} 
Recall that within any slot $n$, uplink NOMA is carried out between server and the corresponding user pair. Then, according to~\eqref{eq:error_tau},
the error probability of  decoding a data packet from user~2  is expressed as
\begin{equation}
     \varepsilon_{2|1,n}={\mathcal{P}}(\gamma_{2|1,n},d,m_n).
\end{equation}
If SIC succeeds, the decoding error probability of user~1 is given by
\begin{equation}
   \varepsilon _{1|1,n}={\mathcal{P}}(\gamma _{1|1,n},d,m_n).
\end{equation}
In the mean time, if SIC fails, user 1 has to decode its own signal with interference, resulting in an error probability of
    $\varepsilon_{1|2,n}={\mathcal{P}}(\gamma_{1|2,n},d,m_n)$.
Recall that $\gamma_{1|2,n}$ is generally less than $\gamma_{2|1,n}$ and also significantly lower than $\gamma_{1|1,n}$. Therefore, the decoding without SIC is unlikely to succeed, i.e., the decoding error probability $\varepsilon_{1|2,n}\approx 1$.
As a result, the overall decoding error probability for user 1 can be written as:
\vspace{-.3cm}
\begin{equation}\vspace{-.2cm}
\begin{split}
    \label{eq:err_1}
    \varepsilon_{1,n}&=(1-\varepsilon_{2|1,n})\varepsilon_{1|1,n}+\varepsilon_{2|1,n}\varepsilon_{1|2,n} 
    \approx \varepsilon_{2|1,n}+ \varepsilon_{1|1,n}.
\end{split}
\end{equation}
And the error probability of decoding the signal of user 2 is straightforward, i.e.,
\vspace{-.2cm}
\begin{equation}\vspace{-.2cm}
    \label{eq:err_2_n}
    \varepsilon_{2,n}
    =\varepsilon_{2|1,n}={\mathcal{P}}(\gamma_{2|1,n},d,m_n).\vspace{-5pt}
\end{equation}
We aim at minimizing the (expected) maximal error probability among the users in each time frame $  \mathbb{E}_t\big[\max\limits_{i,n}\{\varepsilon_{i,n}\}\big]$ by optimally allocating the blocklength of each pair $\mathbf{M}=\{\mathbf{m}(t)|t=1,\dots,T\}$ and transmit power $\mathbf{P}=\{\mathbf{p}(t)|t=1,\dots,T\}$ at any time frame index $t$ while fulfilling the energy consumption budget constraint as shown in Fig~\ref{fig:framework}, i.e.,
\begin{equation}
\sum^N_{n=1}m_n(t)(p_{1,n}(t)+p_{2,n}(t))\leq E_{\max}, \forall t\in\mathcal{T}.
\end{equation} 
{We assume that the transmit power and blocklength can be continuously allocated.}
{In addition, to ensure %
the transmission quality and prevent wasting of radio resources, we can construct a feasible set of $\mathbf{M}$ and $\mathbf{P}$. In particular, on one hand, each user requires the channel to be sufficiently good to satisfy the minimal condition $\gamma_{1|1,n}\geq\gamma_{2|1,n}\geq \gamma_{\rm th}\geq 0$ dB, i.e., $z_{2,n}\geq z_{1,n}\geq z_{\min}\geq  \frac{\sigma^2M}{E_{\max}}$. However, since the channel is random, it is possible that $z_{1,n}\leq z_{\min}$. In such cases, we should allocate no power and no blocklength for that user pair in the corresponding frame, otherwise it may potentially lead to an unfair resource allocation. On the other hand, transmissions with a coding rate greater than Shannon capacity are non-preferred, since it always results in an error probability greater than 0.5. Therefore, we should at least assign $m_n\geq \frac{d}{\log_2(\gamma_{\rm th}+1)}$ to the user pair $n$ to prevent the waste of radio resources. Therefore, we have the feasible set for transmit power allocation $\mathbf{\Omega}_P=\{\mathbf{\Omega}^p_{t}\}^T$ and blocklength allocation $\mathbf{\Omega}_M=\{\mathbf{\Omega}^m_{t}\}^T$, where $\mathbf{\Omega}^p=\{\Omega^p_{n}\}^N$ and $\mathbf{\Omega}^m=\{\Omega^m_{n}\}^N$ is the corresponding allocation in the time frame $t$.
Then, we have the following feasible sets:
\vspace{-5pt}
\begin{equation}
\label{eq:omega_p}
    \Omega^p_{n}=
    \begin{cases}
        p_{2,n}/p_{1,n}\geq \frac{z_{1,n}\gamma_{\rm th}}{z_{2,n}},~p_{1,n}\geq \frac{z_{1,n}\gamma_{\rm th}}{\sigma^2} & \text{if $z_{1,n}\geq z_{\min}$,}\\
        p_{2,n}=p_{1,n}=0, & \text{if $z_{1,n}< z_{\min}$,}
    \end{cases}\vspace{-5pt}
\end{equation}
and\vspace{-5pt}
\begin{equation}
\label{eq:omega_m}
    \Omega^m_{n}=
    \begin{cases}
        m_n\geq \frac{d}{\log_2(\gamma_{\rm th}+1)}, & \text{if $z_{1,n}\geq z_{\min}$,}\\
        m_n=0, & \text{if $z_{1,n}< z_{\min}$.}
    \end{cases}\vspace{-5pt}
\end{equation}
{In addition, to adopt the feasible set, we also modify the error probability with an indicator function as:\vspace{-5pt}
\begin{equation}
    \hat{\varepsilon}_{i,n}=\mathbbm{1}_{z< z_{\min}}(z_{i,n})\varepsilon_{i,n},\vspace{-5pt}
\end{equation}
where $\mathbbm{1}_{z<z_{\min}}(\cdot)$ is the indicator function with condition $z< z_{\min}$.} {In this way, the users, whose channel gains do not satisfy the conditions, also do not influence the value of the maximization $\max_{i,n}\{ \varepsilon_{i,n}\}$ and the number of time slot.}} 
{Moreover, it ensures the accuracy of the approximation we introduced in (2).}
{Then, the corresponding optimization problem can be written as follows:%
\vspace{-10pt}
\begin{mini!}[2]
{_{\mathbf{M}\in\mathbf{\Omega}_M,\mathbf{P}\in\mathbf{\Omega}_P}}{ \mathbb{E}_t\big[\max\limits_{i,n}\{\hat{\varepsilon}_{i,n}\}\big]}
{\label{problem:err_time}}{}\vspace{-20pt}
\addConstraint{\sum^N_{n=1}m_n(t)(p_{1,n}(t)+p_{2,n}(t))\leq E_{\max},\ \forall t\in \mathcal{T}}
\vspace{-10pt}
\label{con:energy_err_t}
\addConstraint{\sum^N_{n=1}m_n(t)\leq M, \forall t\in \mathcal{T}}
\label{con:blocklength_err_t}
\addConstraint{p_{1,n}(t)\leq P_{\max},~p_{2,n}(t)\leq P_{\max}, \ \forall n\in \mathcal{N} ,\ \forall t\in \mathcal{T},}
\label{con:power_err_t}
\vspace{-10pt}
\end{mini!}
where constraint~\eqref{con:blocklength_err_t} indicates that the allocated blocklengths of all TDMA slots should not exceed the maximal available blocklength $M$. Constraint~\eqref{con:power_err_t} restricts the maximal available power for NOMA power allocation. 
However, Problem~\eqref{problem:err_time}  belongs to  dynamic programming, which   is generally challenging to be solved analytically. 
In addition, the joint convexity of the FBL error probability has been recently characterized in OMA networks~\cite{joint_letter}. However, proving the convexity   of Problem~\eqref{problem:err_time} in a NOMA scenario (where power allocation  influences both the numerator and denominator of the SINR simultaneously) is still an open problem. 
In the next subsection, we rigorously prove this convexity and characterize the optimal solution to problem~\eqref{problem:err_time}. 
\vspace{-10pt}
\subsection{Optimal Solution to~\eqref{problem:err_time}}
 {To tackle this issue, we transfer the original problem from the time domain into the channel state domain. In particular, we consider a channel state combination $\mathbf{z}(\tau)=\{z_{1,1}(\tau), \dots,z_{2,N}(\tau)\}$, 
where $\tau\in\mathcal L=\{1,\dots,L\}$ is the index of possible channel realization. Then, we have:\vspace{-10pt}
\begin{equation}
\label{eq:time_to_state}
       \mathbb{E}_t[\varepsilon_{i,n}]
                ={\frac{1}{T}}\sum^T_{t=1}
                    \varepsilon_{i,n}(t) =\mathbb{E}_z[\varepsilon_{i,n}]=\int_{\mathbf{z}}\varepsilon_{i,n}(\tau) f_\mathbf{Z}(\mathbf{z}(\tau))d \mathbf{z}\approx \sum^L_{\tau=1}\varepsilon_{i,n}[\tau] f_\mathbf{Z}(\mathbf{z}[\tau]){\Delta_L}.\vspace{-10pt} 
\end{equation}
}{Recall that we let the user with strong channel gain to be the strong user. Then, $f_{\mathbf{Z}}(\cdot)$ is the joint probability density function (PDF) of sorted channel realization and $\Delta_L$ denotes the resolution for considering $L$ combinations, such that $\sum_{\tau=1}^Lf_\mathbf{Z}(\mathbf{z}[\tau])\Delta_L=1$.}
{Moreover $[\tau]$ indicates that the channel state is discrete at frame $\tau$.}} {Clearly, the approximation becomes accurate as $L\to \infty$. Based on the ergodicity of the channel states, we can replace time frame index $t$ with state frame index $\tau$ in~\eqref{problem:err_time}, resulting in the following equivalent optimization problem:  } 
\begin{mini!}[2]
{_{\mathbf{M}\in\mathbf{\Omega}_M,\mathbf{P}\in\mathbf{\Omega}_P}}{ \mathbb{E}_z\big[\max\limits_{i,n}\{\hat{\varepsilon}_{i,n}\}\big]}
{\label{problem:err_original}}{}
\addConstraint{\vspace{-20pt}\sum^N_{n=1}m_n[\tau]\leq M, \forall \tau \in\mathcal{L}}\vspace{-20pt}
\label{con:blocklength_err_tau}
\addConstraint{\sum^N_{n=1}m_n[\tau](p_{1,n}[\tau]+p_{2,n}[\tau])\leq E_{\max},\ \forall \tau \in\mathcal{L}}\vspace{-10pt}
\label{con:energy_err_tau}
\addConstraint{p_{1,n}[\tau]\leq P_{\max},~p_{2,n}[\tau]\leq P_{\max}, \ \forall n\in \mathcal{N} ,\ \forall \tau \in\mathcal{L}}.
\vspace{-5pt}
\end{mini!}
{
Note that the total available blocklength $M$ is constant in any frame and the energy budget cannot be carried over to the next frame\footnote{{It should be pointed out that the problem is still solvable with the proposed approaches via Lagrange dual method if the energy budget can be carried over. The details is presented in Section IV.}}. Moreover, according to~\eqref{eq:time_to_state}, $\mathbb{E}_t[\hat{\varepsilon}_{i,n}]$ is a linear combination of $\hat{\varepsilon}_{i,n}[\tau]$. In other words, minimizing $\mathbb{E}_t[\hat{\varepsilon}_{i,n}]$ is to minimizing $\hat{\varepsilon}_{i,n}[\tau]$ in each frame $\tau$. Then, we can decompose the problem into $L$ independent sub-problems with the feasible sets. In any arbitrary state $\mathbf{z}[\tau]$, the sub-problem can be written as:}
\begin{mini!}[2]
{_{\mathbf{m}\in\mathbf{\Omega}_m,\mathbf{p}\in\mathbf{\Omega}_p}}{\max\limits_{i,n}\{ \hat{\varepsilon}_{i,n}[\tau]\}}
{\label{problem:err_tau}}{}\vspace{-10pt}
\addConstraint{\sum^N_{n=1}m_n\leq M}\vspace{-5pt} %
\label{con:blocklength}
\addConstraint{\sum^N_{n=1}m_n(p_{1,n}+p_{2,n})\leq E_{\max}}
\label{con:energy}\vspace{-5pt}
\addConstraint{p_{1,n}\leq P_{\max},~p_{2,n}\leq P_{\max}, \ \forall n\in \mathcal{N}}\vspace{-5pt}
\label{con:power}
\addConstraint{\mathbf{z}=\mathbf{z}[\tau].}\vspace{-5pt}
\end{mini!}
Obviously, both the objective function and  constraint~\eqref{con:energy} are non-convex. {To tackle this issue, we introduce a variable substitution $a_n=\sqrt[3]{m_n}$, $b_{1,n}=\frac{1}{p_{1,n}}$, and $b_{2,n}=\frac{p_{1,n}}{p_{2,n}}$.} 
Then,  Problem~\eqref{problem:err_tau} can be further transformed as follows:\vspace{-10pt}
\begin{mini!}[2]
{_{\mathbf{a}\in\mathbf{\Omega}_m,\mathbf{b}\in\mathbf{\Omega}_p}}{\max\limits_{i,n}\{ \hat{\varepsilon}_{i,n}[\tau]\}}
{\label{problem:err_tau_ab}}{}\vspace{-10pt}
\addConstraint{\sum^N_{n=1}a^3_n\leq M}
\label{con:blocklength_err_mp}\vspace{-5pt}
\addConstraint{\sum^N_{n=1}a^3_n\left(\frac{1}{b_{1,n}b_{2,n}}+\frac{1}{b_{1,n}}\right)\leq E_{\max}}
\label{con:energy_err_mp}\vspace{-5pt}
\addConstraint{\frac{1}{b_{1,n}b_{2,n}}\leq P_{\max},~\frac{1}{b_{1,n}}\leq P_{\max}, \ \forall n\in \mathcal{N}}\vspace{-5pt}
\label{con:power_err_mp}
\addConstraint{\mathbf{z}=\mathbf{z}[\tau],}\vspace{-5pt}
\end{mini!}
where the variables $\mathbf a$ and $\mathbf b$ are ,respectively, the vectors including all $a_n$ and $b_{1,n}$, $b_{2,n}$. 
Then, to solve Problem~\eqref{problem:err_tau_ab}, we have following Lemma characterizing the joint convexity.\vspace{-5pt}
\begin{theorem}
\label{lemma:err_convex}
Problem~\eqref{problem:err_tau_ab} is a convex problem.
\vspace{-5pt}
\end{theorem}
\vspace{-5pt}
\begin{proof}
Appendix A.
\end{proof}
\textbf{Remark 1:}
\emph{
{The characterized joint convexity of Lemma 1 can also be applied for pure TDMA or OMA scheme by either fixing the transmit power $\mathbf{p}$ or blocklength $\mathbf{m}$ without the variable substitution, i.e., $\varepsilon_{i,n}$ is convex in either $\mathbf{p}$ or $\mathbf{m}$. }
}

Based on Lemma~\ref{lemma:err_convex}, Problem~\eqref{problem:err_tau_ab} can be solved efficiently via any standard convex optimization tools. Therefore, the optimal solutions of original Problem~\eqref{problem:err_original} can be obtained in the following approach: 
{In any time frame $t$, we let $\mathbf{z}[\tau]=\mathbf{z}(t)$ and solve Problem~\eqref{problem:err_tau_ab}, resulting in the optimal solutions $\mathbf{a}^*[\tau]$ and $\mathbf{b}^*[\tau]$. Subsequently, we obtain the optimal blocklength and power allocation by reversing the variable substitution with $m^*_n(t)=\left(a^*_n[\tau]\right)^3$, $p^*_{1,n}(t)=\frac{1}{b^*_{1,n}
[\tau]}$ and $p^*_{2,n}(t)=\frac{1}{b^*_{1,n}[\tau]b^*_{2,n}[\tau]}$, $\forall n\in\mathcal{N}$. 
Although this approach requires solving Problem~\eqref{problem:err_tau_ab} for all possible $t$, we only need to solve for the current time frame $t$ based on any instantaneous channel realization $\mathbf{z}(t)$ in a practical system, with a low computational complexity of $\mathcal{O}(4N^2)$.} 

{Recall that the solutions are obtained based on SINR approximation in~\eqref{eq:snr_app}, which may not be accurate in every time frame. Therefore, the optimal results should be recalculated based on the exact SINR expression with optimal solutions $\mathbf{m}^*$ and $\mathbf{p}^*$, while the optimal results based on the approximation can be considered as a performance lower bound. {More importantly, we can show that the results obtained via our approach can achieve a nearly global optimum, as observed in the numerical results in Section V. We also provide a summary of the proposed algorithm in Algorithm 1.}}

{Recall that there are two users in each time slots. However, it is also possible to schedule multiple users to operate the uplink NOMA, where our proposed algorithm can be extended. In particular, we can leverage the block coordinate descent (BCD) method~\cite{BCD} to iteratively solve Problem (19) by fixing one of the transmit power of those users. However, the NOMA performs better with two users in practice, while the performance of such schemes with more than two users are heavily influenced by the interference and error propagation~\cite{SIC_NOMA}. Therefore, in the rest of the paper, we focus on the two-user case.}

\begin{algorithm}[!t]
\scriptsize
\caption{Algorithm to solve (16)}\label{alg:err}
\begin{algorithmic}[1]
\State \textbf{Initial:} $\mathbf{z}[\tau]=\mathbf{z}(t)$ 
\For {user pair $n=1,\dots,N$}
\If {$z_{i,n}<z_{\min}$}
    $m^*_{n}=0,p^*_{1,n}=0,p^*_{2,n}=0$
\EndIf
\EndFor 
\Ensure (11), (12), (18b), (18c), and (18d)
\State Solve (18) according to Lemma 1 and get $(\mathbf{a}^*[\tau],\mathbf{b}^*[\tau])$
\State Reverse the variable substitution with $m_n^*[\tau] =(a_n^*[\tau])^3$ $p_{1,n}^*[\tau] = 1/b_{1,n}^*[\tau]$ and $p_{2,n}^*[\tau] = 1/(b_{1,n}^*[\tau]b_{1,n}^*[\tau])$.
\State Calculate the exact SINR with $\gamma_{2|1,n}[\tau] = \frac{z_{2,n}[\tau]p^*_{2,n}[\tau]}{z_{1,n}[\tau]p^*_{1,n}[\tau]+\sigma^2_n}$, $\forall n\in\mathcal{N}$.
\State  Reconstruct the optimal results with $\pmb{\varepsilon}^*[\tau] =  {\hat{\pmb{\varepsilon}}}(\mathbf{m}^*_{i,n}[\tau],\mathbf{p}^*[\tau],\mathbf{z}[\tau])$.
\end{algorithmic}
\end{algorithm}

\section{Effective Capacity-Oriented Framework Design}
\label{sec:EC}

{In the previous section, the transmitted data packet in each frame for each user is deterministic.} 
It should be pointed out that  in certain practical applications 
data arrives continuously and stays in the queue buffer until being transmitted. 
For instance, these applications include video streaming in virtual/augmented/mixed reality. 
{In such queue-influenced scenarios, the applications are not only delay-sensitive, but also heavily influenced by the queuing behavior.}  
Therefore, a pure physical-layer  metric   is not sufficient for characterizing the performance of such systems. 
Hence, in this section, we are motivated to adopt the well-known link-layer QoS performance model, namely, effective capacity. 
Note that in our scenario, the system also takes the impact of FBL codes into account. In other words,  the applied effective capacity model indicates the successfully transmitted and decoded data throughput per channel use in the FBL regime, whose queue delay is satisfying specific statistical QoS requirements. %
In the following subsections,  we provide a design framework maximizing the effective capacity while guaranteeing the targeted QoS requirements. 
In particular, we first derive the effective capacity of the considered network in the FBL regime. Subsequently, we state our optimization problem and leverage the Lagrangian dual method to transfer the problem into several solvable subproblems. Finally, we provide the optimal solution.

\vspace{-10pt}
\subsection{Effective Capacity Characterization  and Problem Formulation}
Consider a  target set of QoS requirements \{$\theta_{i,n}$,~$\bar \varepsilon_{i,n}$\}, which are the target
error probability   and  target QoS exponent   for each user $i$ in the $n^{\rm th}$ user pair. 
Recall that the channels are assumed to experience block-fading. Therefore, the transmission error is independent between time frames. 
Then, the effective capacity also influences by those transmission errors. 
{In particular, the discrete service rate in this scenario $s_{i,n}(t)$  becomes the successfully transmitted  bits at time frame $t$. 
In particular, for any user, if the transmission succeeds, i.e., with probability of $(1-\bar{\varepsilon}_{i,n})$, the service process $s_t$ at time frame $t$ is $m_nr_{i,n}$.} By substituting $s_t$ in~\eqref{eq:EC_original}, the effective capacity (in bits/frame) for user $i$ in that pair with FBL codes is given by~\cite{gursoy_EC}:
\vspace{-5pt}
\begin{equation}
    \label{eq:effective_capacity}
    R_{i,n}=%
    -\frac{1}{\theta_{i,n}}
            \ln \{
                    \mathbb{E}[e^{-\theta_{i,n}m_{n}r_{i,n}}
                    (1-\bar\varepsilon_{i,n})+\bar\varepsilon_{i,n}]
                \} ,  
\end{equation}
where $r_{i,n}$ is the coding rate. 
{Compared to the original expression in~\eqref{eq:EC_original},~\eqref{eq:effective_capacity} indicates the influence of both queue delay and decoding error probability.} {Next, let us investigate the coding rate in NOMA.} In particular, the coding rate for user 1 and user 2 in the $n^{\rm th}$ pair can be written as:
\vspace{-5pt}
\begin{equation}
    \label{eq:coding_rate_1}
    r_{1,n}\approx\mathcal{C}(\gamma_{1|1})-\sqrt{\frac{V(\gamma_{1|1})}{m_n}}Q^{-1}(\bar\varepsilon_{i,n}),
\end{equation}
and
\vspace{-5pt}
\begin{equation}
    \label{eq:coding_rate_2}
    r_{2,n}\approx\mathcal{C}(\gamma_{2|1})-\sqrt{\frac{V(\gamma_{2|1})}{m_n}}Q^{-1}(\bar\varepsilon_{i,n}),
\end{equation}
where we still assume $z_2\geq z_1$. 
According to~\cite{FBL_EC}, the  above approximations are accurate in the considered scenarios with reliable transmissions. 
We aim at maximizing the (normalized) sum effective capacity, i.e., $\frac{1}{M}\sum^N_{n=1}\sum^2_{i=1}R_{i,n}$ by optimally allocating the blocklength of each pair $\mathbf{M}=\{\mathbf{m}(t)|t=1,2,\dots\}$ and transmit power of each user $\mathbf{P}=\{\mathbf{p}(t)|t=1,2,\dots\}$ while guaranteeing  QoS conditions. Moreover, since the transmitted data size is not fixed, instead of the deterministic energy budget $E_{\max}$, we consider an average constraint with average energy budget $\bar{E}_{\max}$, i.e.,
\begin{equation}
    \mathbb{E}\Big[\sum^N_{n=1}m_n(t)(p_{1,n}(t)+p_{2,n}(t))\Big]\leq \bar{E}_{\max}.
\end{equation} 
Therefore, we can formulate following optimization problem:
\vspace{-10pt}
\begin{maxi!}[2]
{_{\mathbf{M}\in\mathbf{\Omega}_M,\mathbf{P}\in\mathbf{\Omega}_P}}{\frac{1}{M}\sum^N_{n=1}\sum^2_{i=1}R_{i,n}}
{\label{problem:effective_capacity_time}}{}
\addConstraint{\sum^N_{n=1}m_n(t)\leq M, \forall t\in \mathcal{T}}
\label{con:blocklength_EC_t}
\addConstraint{\mathbb{E}[\sum^N_{n=1}m_n(t)(p_{1,n}(t)+p_{2,n}(t))]\leq \bar{E}_{\max}}
\label{con:energy_EC_t}
\addConstraint{p_{1,n}(t)\leq P_{\max},~p_{2,n}(t)\leq P_{\max}, \ \forall n\in \mathcal{N} ,\ \forall t\in \mathcal{T}}
\label{con:power_EC_t}
\addConstraint{\gamma_{i|j,n}(t)\geq \gamma_{\rm th}, \ \forall n\in \mathcal{N},\ \forall (i,j)\in\{1,2\},\ \forall t\in \mathcal{T}.}
\label{con:snr_EC_t}
\end{maxi!}
Although Problem~\eqref{problem:effective_capacity_time} has a structure similar to that of Problem~\eqref{problem:err_time}, the methodology of solving~\eqref{problem:err_time} cannot be simply followed. Firstly, the average energy budget constraints cannot be decomposed into sub-constraint since the energy consumption of each frame depends on each other. Secondly, the proof of joint concavity for effective capacity is non-trivial, since it has a more complicated expression than error probability according to~\eqref{eq:effective_capacity}. Therefore, in what follows, we exploit Lagrange multipliers and  decompose the original problem into the corresponding partial dual problems. Then, we show that the strong duality holds and provide the optimal solutions via decomposing the dual problem into sub-problems. Finally, we solve the dual problem after characterizing the convexity of any sub-problem with the associated Lagrange multiplier.

\subsection{Optimal Solution to~\eqref{problem:effective_capacity_time}}
In particular, we also consider a channel state combination $\mathbf{z}[\tau]=\{z_{1,1}[\tau],\dots,z_{2,N}[\tau]\}$, for $\tau\in\mathcal L=\{1,\dots,L\}$, with $L\to \infty$. Then, the (expected) effective capacity over time is equivalent to the one averaged over channel states:
\vspace{-5pt}
\begin{equation}
       R_{i,n}=-\frac{1}{\theta_{i,n}}
            \ln \{
                    \mathbb{E}_t[e^{-\theta_{i,n}m_{n}r_{i,n}}
                    (1-\bar\varepsilon_{i,n})+\bar\varepsilon_{i,n}]
                \}
                =-\frac{1}{\theta_{i,n}}
            \ln \{
                    \mathbb{E}_\tau[e^{-\theta_{i,n}m_{n}r_{i,n}}
                    (1-\bar\varepsilon_{i,n})+\bar\varepsilon_{i,n}]
                \}.
\end{equation}
We can apply the same approach on the averaged energy budget constraint:
\vspace{-5pt}
\begin{equation}
\mathbb{E}_t\Big[\sum^N_{n=1}m_n(t)(p_{1,n}(t)+p_{2,n}(t))\Big]\leq \bar{E}_{\max}\iff \mathbb{E}_\tau\Big[\sum^N_{n=1}m_n[\tau](p_{1,n}[\tau]+p_{2,n}[\tau])\Big]\leq \bar{E}_{\max}.
\end{equation}

Then, we  replace time frame index $t$ with state frame index $\tau$ in~\eqref{problem:effective_capacity_time}, resulting in the following equivalent optimization problem:     
\begin{maxi!}[2]
{_{\mathbf{M}\in\mathbf{\Omega}_M,\mathbf{P}\in\mathbf{\Omega}_P}}{\frac{1}{M}\sum^N_{n=1}\sum^2_{i=1}R_{i,n}}
{\label{problem:effective_capacity_original}}{}
\addConstraint{\sum^N_{n=1}m_n[\tau]\leq M, \forall \tau \in\mathcal{L}}
\label{con:blocklength_EC_tau}
\addConstraint{\mathbb{E}_{\mathbf{z}}[\sum^N_{n=1}m_n[\tau](p_{1,n}[\tau]+p_{2,n}[\tau])]\leq E_{\max}}
\label{con:energy_EC_tau}
\addConstraint{p_{1,n}[\tau]\leq P_{\max},~p_{2,n}[\tau]\leq P_{\max}, \ \forall n\in \mathcal{N} ,\ \forall \tau \in\mathcal{L}}
\label{con:power_EC_tau}
\addConstraint{\gamma_{i|j,n}[\tau]\geq \gamma_{\rm th}, \ \forall n\in \mathcal{N},\ \forall (i,j)\in\{1,2\},\ \forall \tau \in\mathcal{L}.}
\label{con:snr_EC_tau}
\end{maxi!}
However, we cannot decompose the above problem directly due to the constraint~\eqref{con:energy_EC_tau}. Instead, we apply the Lagrange dual method to obtain the Lagrangian function:
\vspace{-5pt}
\begin{equation}
    \begin{split}
    L=&-\frac{1}{M}\sum^N_{n=1}\sum^2_{i=1}R_{i,n}+\lambda_E(\mathbb{E}_{\mathbf{z}}[\sum^N_{n=1}m_n[\tau](p_{1,n}[\tau]+p_{2,n}[\tau])]-E_{\max})%
    \end{split}
\end{equation}
where $\lambda_E$ is the Lagrange Multiplier for constraint~\eqref{con:energy_EC_tau}. With $L$, the corresponding dual problem is given by:
\begin{mini!}[2]
{\lambda_E\geq 0}{\underset{\mathbf{M},\mathbf{P}}{\inf} L(\mathbf{M},\mathbf{P},\lambda_E)}
{\label{problem:dual}}{}
\end{mini!}
 {Note that Problem~\eqref{problem:effective_capacity_original} is not necessarily convex, i.e.,  we cannot apply Slater's condition directly. To address this issue, we have the following lemma:}
\begin{theorem}
The time-sharing condition~\cite{dual} is satisfied for Problem~\eqref{problem:effective_capacity_original}, if it holds that $L\to\infty$.
\label{lemma:time_sharing}
\end{theorem}
\begin{proof} Appendix~B.
\end{proof} 
{Therefore, the strong duality holds and we can obtain the optimal solutions by solving the dual problem~\cite{intro_NOMA_3,dual}.}
Then, for a given $\lambda_E$, which is independent from $\tau$, the dual problem can be decomposed into $L$ sub-problems with any state $\tau$, i.e.,
\begin{mini!}[2]
{_{\mathbf{m}\in\mathbf{\Omega}_m,\mathbf{p}\in\mathbf{\Omega}_p}}{L[\tau]= -\frac{1}{M}\sum^N_{n=1}\sum^2_{i=1}R_{i,n}[\tau]+\lambda_E(\sum^N_{n=1}m_n[\tau](p_{1,n}[\tau]+p_{2,n}[\tau])-E_{\max})}
{\label{problem:sub}}{}
\addConstraint{\sum^N_{n=1}m_n[\tau]\leq M}
\label{con:blocklength_tau}
\addConstraint{p_{1,n}[\tau]\leq P_{\max}, p_{2,n}[\tau]\leq P_{\max}}
\label{con:power_tau}
\addConstraint{\mathbf{z}=\mathbf{z}[\tau],}
\end{mini!}
where $\mathbf{\Omega}_{m}$ and $\mathbf{\Omega}_{p}$ are the feasible sets defined in~\eqref{eq:omega_m} and~\eqref{eq:omega_p}, respectively. It should be emphasized that the objective function itself in the dual problem is not decomposed, which is different from the approach in the previous section. Clearly, the dual problem consists of two components: the effective capacity of each user $R_{i,n}$ at the state $\tau$ and the energy budget constraint associated with $\lambda_{\rm E}$.
Instead of dealing with $R_{i,n}$ as a function of blocklength and transmit power directly, we characterize the transmission rate in terms of blocklength and S(I)NR. Let $\boldsymbol{\gamma}=\{\gamma_{2|1,n},\gamma_{1|1,n}|n\in\mathcal{N}\}$ denote the S(I)NR matrix associated with the feasible transmit power $\mathbf{p}$. Although the convexity of $r_{i,n}$ with respect to a single factor, i.e., with respect to either $\mathbf{m}$ or $\boldsymbol{\gamma}$ has already been proven~\cite{Hu_TWC_2015}, the feature of joint convexity has not been characterized yet. To address this, we establish the following lemma:
\begin{theorem}
\label{lemma:r_convex}
Negative transmission rate $-r_{i,n}$ is jointly convex in $\mathbf{m}$ and $\boldsymbol{\gamma}$ within the feasible set of Problem~\eqref{problem:sub}.
\end{theorem}
\begin{proof}
Appendix~C.
\end{proof} 

Lemma~\ref{lemma:r_convex} can already be applied to solve the problem with transmission rate as the objective function. Moreover, with the analytical results of Lemma~\ref{lemma:r_convex}, we can further characterize the joint convexity of negative effective capacity $-R_{i,n}$, since $-R_{i,n}$ can be considered as a function of $m_n$ and $-r_{i,n}$. 
It should be pointed out that the vector composition rule for proving the joint convexity cannot be directly applied, since it requires $-r_{i,n}(m_n,\gamma_{i,n})$ to be non-increasing and concave. Therefore, we provide the following lemma by exploiting the sign of Hessian matrix of $-R_{i,n}$:\vspace{-10pt}
\begin{theorem}
\label{lemma:EC_convex}
Negative effective capacity $-R_{i,n}(m_n,r_{i,n})$, $\forall i,n$, is jointly convex in $\mathbf{m}$ and $\boldsymbol{\gamma}$.\vspace{-5pt}
\end{theorem}\vspace{-10pt}
\begin{proof}
Appendix~D.\vspace{-5pt}
\end{proof}
However, the objective function of the dual problem is still not convex due to the average energy budget constraint. To tackle this, we also replace the variables as $a_n=m^3_n$, $b_{1,n}=\frac{1}{p_{1,n}}$, $b_{2,n}=\frac{p_{1,n}}{p_{2,n}}$, $\forall n\in\mathcal{N}$, resulting in following optimization problem:\vspace{-5pt}
\begin{mini!}[2]
{_{\mathbf{b}\in\mathbf{\Omega}_m,\mathbf{a}\in\mathbf{\Omega}_p}}{L[\tau]= -\frac{1}{M}\sum^N_{n=1}\sum^2_{i=1}R_{i,n}[\tau]+\lambda_E(\sum^N_{n=1}a^3_n[\tau](1/b^2_{1,n}[\tau]+1/(b^2_{1,n}[\tau]b^2_{2,n}[\tau]))-E_{\max})}
{\label{problem:sub_ab}}{}\vspace{-5pt}
\addConstraint{\sum^N_{n=1}a^3_n[\tau]\leq M}
\label{con:blocklength_sub_ab}\vspace{-5pt}
\addConstraint{\frac{1}{b_{1,n}[\tau]}\leq P_{\max},~\frac{1}{b_{1,n}[\tau]b_{2,n}[\tau]}\leq P_{\max}}
\label{con:power_sub_ab}\vspace{-5pt}
\addConstraint{\mathbf{z}=\mathbf{z}[\tau].}\vspace{-5pt}
\end{mini!}\vspace{-5pt}
Combing all of the above analytical results, {i.e., Lemma 2 and Lemma 3}, we have the following lemma characterizing the corresponding convexity:\vspace{-5pt}
\begin{theorem}\vspace{-5pt}
\label{lemma:EC_convex_problem}
Problem~\eqref{problem:sub_ab} is a convex problem.\vspace{-5pt}
\end{theorem}\vspace{-5pt}
\begin{proof}
Following the methodology  in Lemma 1, we can prove that the Hessian matrix of~\eqref{problem:sub_ab} is positive semi-definite. In particular, with
 $\frac{\partial m_n}{\partial a_n}=3a^2_n$, $\frac{\partial^2 m_n}{\partial a^2_n}=6a_n$, $\frac{\partial \gamma_{1,n}}{\partial b_{1,n}}=-\frac{z_1}{b^2_{1,n}\sigma^2}$,
 $\frac{\partial \gamma_{1,n}}{\partial b_{1,n}}=\frac{2z_1}{b^3_{1,n}\sigma^2}$ and
 $\frac{\partial \gamma_{2,n}}{\partial b_{2,n}}=\frac{z_1}{z_2}$, we can show that the joint convexity of $-R_{i,n}$ still holds, while $\lambda_E(\sum^N_{n=1}a^3_n(1/b^2_{1,n}+1/(b^2_{1,n}b^2_{2,n}))-E_{\max})$ is also jointly convex for any $\lambda_E\geq 0$ via checking the determinate of Hessian matrix $\frac{6a^7_n(b_{1,n}+1)}{b^5_{2,n}b^4_{1,n}}\geq 0$. The proof is quite standard as shown in Lemma 1, and to avoid repetition, we omit the details of the proof. 
Thus, as the sum of convex functions, $L[\tau]$ is also jointly convex. In the meantime, both constraints~\eqref{con:blocklength_sub_ab} and~\eqref{con:power_sub_ab} are convex. Hence, Problem~\eqref{problem:sub_ab} is a convex problem.
\end{proof}

\vspace{-5pt}
\textbf{Remark 2:} \emph{
Lemma~\ref{lemma:EC_convex_problem} can also be generalized to solve the optimal blocklength allocation for pure OMA scheme , as well as the optimal power allocation for pure NOMA scheme. In fact, they can be considered as special cases of Lemma 2 with either fixed blocklength $\mathbf{m}$ or transmit power $\mathbf{p}$.} 

Based on~Lemma~\ref{lemma:EC_convex_problem}, Problem~\eqref{problem:sub_ab} can be solved efficiently with standard convex optimization tools for a given $\lambda_E$ with any state $\tau$. 
{Therefore, all sub-problems can be solved via all channel realizations in parallel independently with a sufficient large number of $L$. Then, we reconstruct the optimal solutions for Problem~\eqref{problem:sub} with $m_n^*=(a^*_n)^3$, $p_{1,n}^*=\frac{1}{b_{1,n}^*}$ and $p_{2,n}^*=\frac{1}{b_{2,n}^*b_{1,n}^*}$, $\forall n\in\mathcal{N}$. Therefore, the dual problem~\eqref{problem:dual} can also be solved iteratively with the sub-gradient method, where the updated sub-gradient is $(E_{\max}-\mathbb{E}[\sum^N_{n=1}m^*_n(p^*_{1,n}+p_{2,n}^*)])$~\cite{Boyd}.  
Then, with solution of the dual problem $\lambda^*_{\text{E}}$, the optimization problem in any time frame $t$ can be solved efficiently.} 
 
{Note that there can exist multiple set of solutions $\mathbf{M}^*$ and $\mathbf{P}^*$. However, it is possible that some solutions may not satisfy the constraints of Problem~\eqref{problem:effective_capacity_original}~\cite{dual}. To tackle this issue, after obtaining $\mathbf{M}^*$ and $\mathbf{P}^*$, we should always check whether the inequalities (29b)-(29e) hold. It should be emphasized that if the time-sharing condition is satisfied, there always exists a set of feasible $\mathbf{M}^*$ and $\mathbf{P}^*$. A corresponding pseudo-code is shown in Algorithm 2.} 

\begin{algorithm}[!t]
\scriptsize
\caption{Algorithm to solve~\eqref{problem:effective_capacity_original}}\label{alg:EC}
\begin{algorithmic}[1]
\State Generate a set of channel realization $\mathbf{z}$ with a sufficiently large $L$.
\For {$\tau=1,\dots, L$}
    \State \textbf{Initial:} $\mathbf{z}=\mathbf{z}[\tau]$ 
        \For {user pair $n=1,\dots,N$}
            \If {$z_{i,n}<z_{\min}$}
                $m^*_{n}=0,p^*_{1,n}=0,p^*_{2,n}=0$
            \EndIf
        \EndFor
    \Ensure (14), (15), (29b), (29c), and (29d)
    \State Solve (29) according to Lemma 1 and get $(\mathbf{a}^*[\tau],\mathbf{b}^*[\tau])$
    \State Reverse the variable substitution with $m_n^*[\tau] =(a_n^*[\tau])^3$ $p_{1,n}^*[\tau] = 1/b_{1,n}^*[\tau]$ and $p_{2,n}^*[\tau] = 1/(b_{1,n}^*[\tau]b_{1,n}^*[\tau])$.
\EndFor
\State Update $\lambda_{\text{E}}$ according to (31) based on  $\mathbf{M}^*$ and $\mathbf{P}^*$.
\If{ The average energy constraint (29c) satisfied with $\mathbf{M}^*$ and $\mathbf{P}^*$} 
\State  $\lambda^*_{\text{E}}=\lambda_{\text{E}}$.
    \Else
        \State let $\tau=1$ and return to Step 2.
\EndIf
\end{algorithmic}
\end{algorithm}

\vspace{-5pt}
\section{Numerical Results}
In this section, we provide the numerical results to validate our analytical characterizations and evaluate the performance of both the proposed reliability-oriented and effective capacity-oriented designs. To demonstrate  the advantage of our approaches, we also show the performance of other benchmarks under the same setups. In the following section, we provide the simulation parameters and discuss the details of the applied benchmarks. Then, we present the corresponding simulation results.
\vspace{-10pt}
\subsection{Simulation and Benchmark Setups}
{Unless specifically mentioned, we adopt the simulation setup from~\cite{Hu_reference}. The default parameterization is as follows: Maximal transmit power $P_{\max}=30$ dBm for each user,} {where up to $5$ user pairs are available in the network, i.e., $10$ users in total with 5 time slots and 10 power levels in our hybrid NOMA-TDMA scheme.} We set a unit average channel gain for all links. Furthermore, those channels are assumed to experience i.i.d. block Rayleigh fading, i.e., $z_{i,n}\sim\mathcal{N}\{0,1\}$, $\forall i,n$. {For the reliability-oriented design, we consider  that transmissions with packet size of $d=320$ bits are carried out with default total blocklength $M=350n$ in each time frame with energy budget $E_{\max}= 700$ J. Meanwhile, we set a unified target error probability as $\bar \varepsilon = 10^{-4}$ and QoS exponent $\theta=10^{-3}$ with average energy budget $\bar{E}_{\max}=700$ J for the effective capacity-oriented design.} {The simulations are carried out with 2000 time frames.} Moreover, we also consider the following two approaches as benchmarks:
\begin{itemize}
    \item \textbf{{Hybrid NOMA-TDMA with the IBL solutions (Hybrid IBL)}}: 
    {Consider the exact same hybrid NOMA-TDMA scheme in this paper, but the framework is constructed based on the ideal infinite blocklength (IBL)  assumptions, i.e., transmissions are arbitrarily reliable at Shannon's capacity.  In particular, the blocklength is uniformly distributed to each user pair with $m_n=\frac{M}{n}$ unless the coding rate is above the Shannon capacity. Moreover, the power allocation is adopted from the widely applied strategy in~\cite{NOMA_power_IBL}.} {The FBL performance with such IBL solutions is shown to demonstrate the importance of investigating the URLLC performance.}
    \item {\textbf{Optimal OMA scheme (TDMA optimal)}: Consider a pure OMA scheme in the time domain with a single carrier, where the blocklength (in symbols) is divided into $2n$ slots for each user.} However, instead of uniform distribution, the blocklength is optimally allocated according to Corollary 1. It should be pointed out that the energy budget $E_{\max}$ (or average energy budget $\bar{E}_{\max}$) should still be fulfilled in the considered scheme.  The performance of this scheme is illustrated to show the advantage of introducing the hybrid scheme compared to the conventional OMA schemes.
\end{itemize}
\vspace{-10pt}
\subsection{Reliability Comparison}
\vspace{-5pt}
{We first show the reliability comparison between our proposed design and aforementioned benchmarks under two setups:} $i)$ The total available blocklength is fixed regardless of the number of user pairs, i.e., $M=700$ symbols. $ii)$ The average available blocklength remains consistent regardless of the number of user pairs, i.e., $M=350n$ symbols. Recall that we leverage the SINR approximation in~\eqref{eq:snr_app} to characterize the optimization problems. Therefore, the simulations also depict the  
results based on this approximation, referred to as \textbf{hybrid approx.}, as well as the results obtained via exhaustive search, referred to as \textbf{hybrid exhaustive}. Finally, the performance results of our approaches as depicted as \textbf{hybrid optimal}.

\begin{figure}[!t]
\begin{minipage}[t]{0.49\linewidth}
    \centering
\includegraphics[width=0.8\textwidth,trim=10 20 10 10]{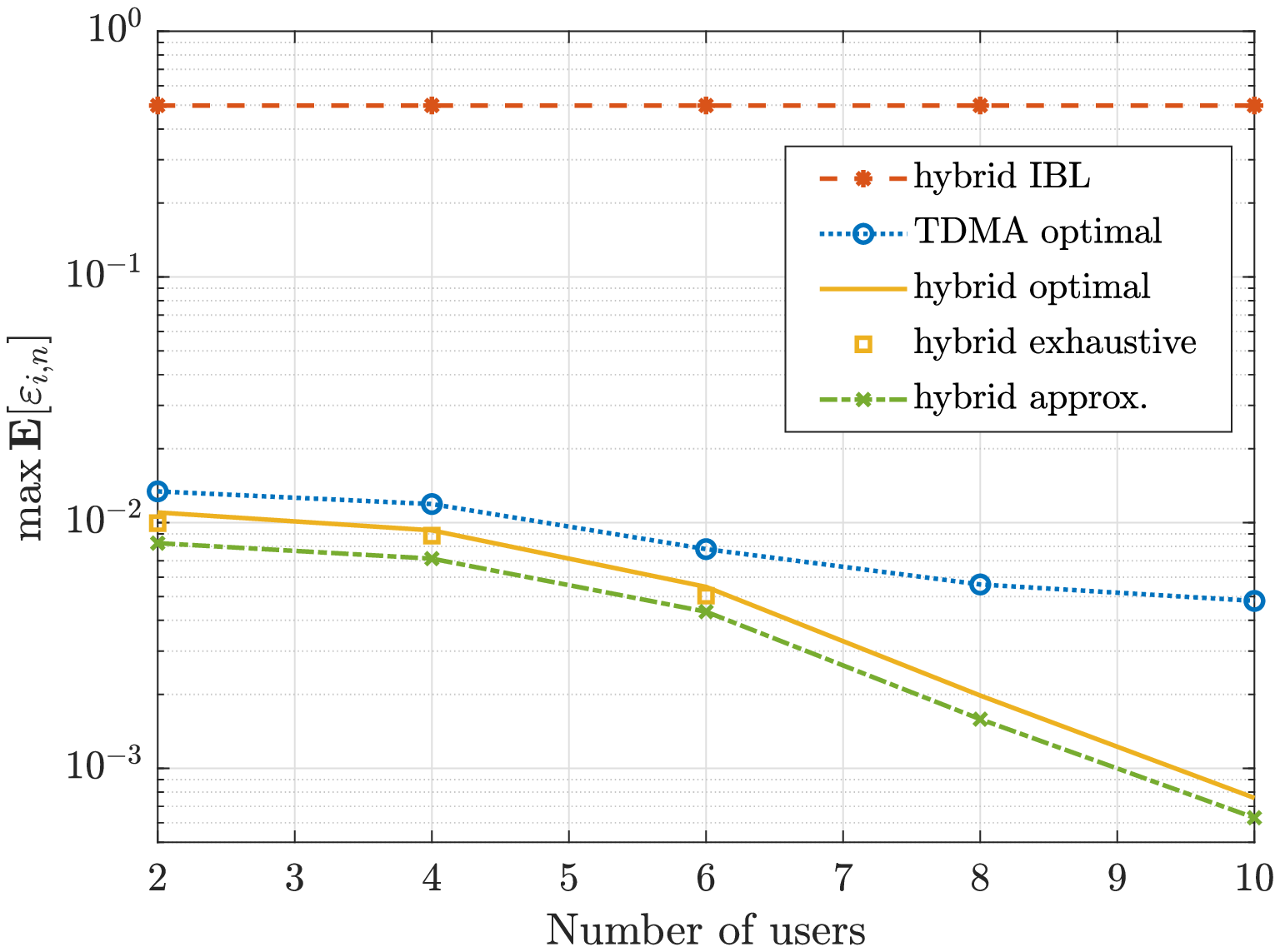}
\caption{Error Probability evaluation  in comparison to benchmarks with an adaptive total blocklength $M=350n$ symbols.}
\label{fig:err_comparison_adaptive}
\end{minipage}
\begin{minipage}{0.005\linewidth}
    ~
\end{minipage}
\begin{minipage}[t]{0.49\linewidth}
    \centering
\includegraphics[width=0.805\textwidth,trim=10 20 10 10]{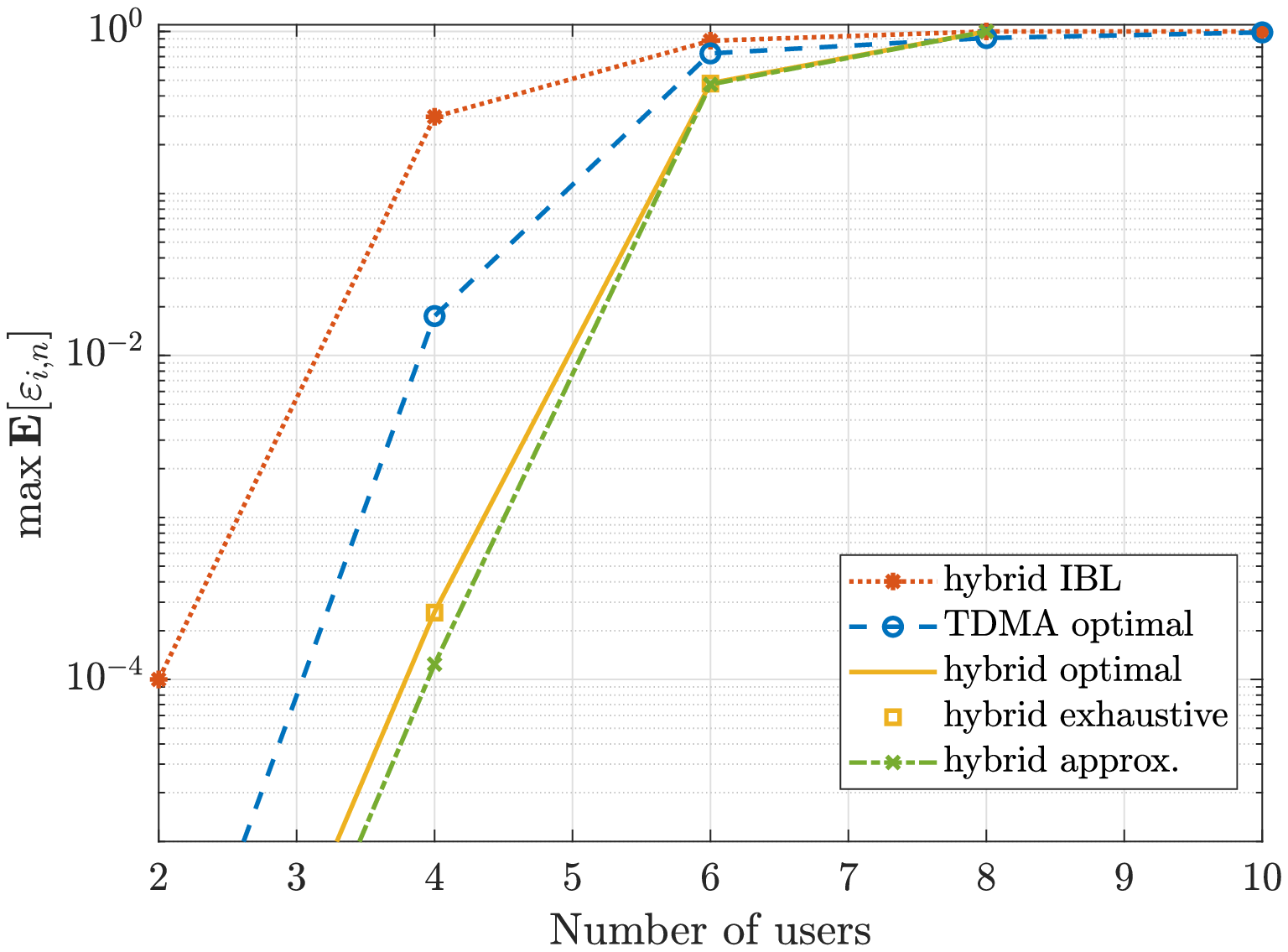}
\caption{Error Probability evaluation  in comparison to benchmarks with  a fixed total blocklength $M=700$ symbols.}
\label{fig:err_comparison_fix}
\end{minipage}
\vspace{-26pt}
\end{figure}

In particular, we plot the optimal average maximal error probability
against the number of user in Fig.~\ref{fig:err_comparison_adaptive}. To evaluate the performance for different network scales, we set the total blocklength as $m=350n$ symbols, i.e., each user pair has an average blocklength of 350 symbols. In such a scenario, the error probability decreases with increasing number of users in our hybrid NOMA-TDMA design. This is due to the fact that our design is able to fully utilize the radio resources between users to maintain the error probability balance among the users with channel differences. For instance, if one user has a low channel gain in an arbitrary channel realization, as long as the channel gain is not below the threshold $z_{\min}$, the system will assign more blocklength and transmit power to that user so that the overall error probability is improved. 

{It should be emphasized that the performance improvement is achieved via our analytical results that explicitly take the FBL impact into account. If we simply adopt the IBL assumption, the performance will be much worse, as shown in hybrid IBL. Since the power and blocklength allocation in the hybrid IBL scheme are fixed based on Shannon capacity, increasing the number of users does not influence the performance at all.} On the other hand, if the system solely relies on OMA schemes, e.g., TDMA, it may be challenging to support higher connectivity while guaranteeing the performance. 
{This can be observed by comparing the performance of our proposed scheme and pure TDMA scheme. 
In particular, the hybrid NOMA-TDMA scheme performs slightly better than TDMA when the user number is low. However, when the user number increases, the advantage of our proposed scheme enlarges significantly. This is due to the fact that the NOMA scheme benefits from the diversity of the users.}

Overall,     the above results indicate   the flexibility of the proposed scheme. Next, we again plot the error probability as a function of the number of users, but with fixed total available blocklength $M=700$ symbols. Clearly, since all users have to share the available fixed blocklength, the error probability increases with increasing number of users regardless of which scheme is used. However, our hybrid design still shows better resilience until the system is overloaded. 
{However, our hybrid design still shows better resilience until the system is overloaded. In such case, the system should either relax the latency requirements or introduce more available carriers to support higher user number.} 
It is worth mentioning that our design has a higher error probability numerically in comparison with hybrid IBL after 8 users in the system. However, this does not mean that our design has degraded. Instead, it simply implies when under that setup, the optimization problem~\eqref{problem:err_tau_ab} becomes infeasible for most channel realizations, where we consider the error probability of users in those channel realizations as being equal to one. 
{Similar to Fig. 3, the performance of the optimized TDMA scheme is actually acceptable when the radio resource is sufficient. Nevertheless, it is still outperformed by our design for all number of users. This comparison demonstrates the resilience of our hybrid NOMA-TDMA scheme.}

In both of figures, we also illustrate the performance of the SINR approximation in~\eqref{eq:snr_app}, as well as the results of exhaustive search, which can be considered as practical performance lower bounds. It is observed that the approximation does introduce some performance gap, but the trends of the curves remain the same. Moreover, the optimal value of error probability $\varepsilon^*_{i,n}$ is actually calculated based on the exact SINR, while the optimal solutions of $\mathbf{m^*}$ and $\mathbf{p^*}$ are obtained via the approximation. Therefore, the actual performance gap is more insignificant, as shown in \textbf{hybrid optimal} and \textbf{hybrid exhaustive}. It should be pointed out that the results of exhaustive search are only available until 6 users, since the computational complexity for larger number of users is too high. This observation further confirms the scalability and efficiency of our proposed approaches.

\begin{figure}[!t]
\begin{minipage}[t]{0.49\linewidth}
    \centering
\includegraphics[width=0.8\textwidth,trim=20 25 25 15]{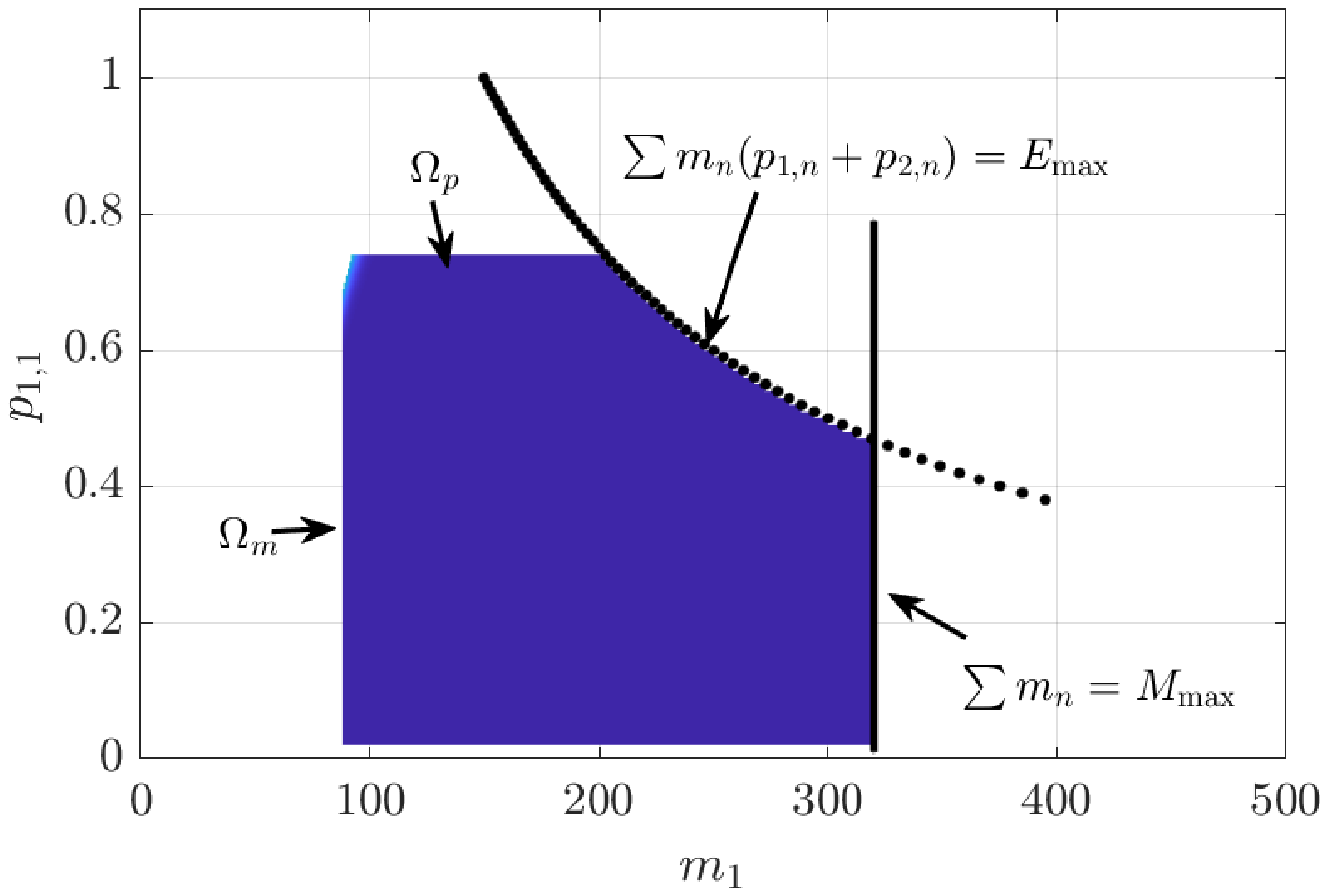}
\caption{Feasible set of~\eqref{problem:err_tau} with 4 users and channel realization $\mathbf{z}=\{0.5,1.5;0.5,1.5\}$, where $m_1$ and $p_{1,1}$ are presented as variables.}
\label{fig:feasible_mp}
\end{minipage}
\begin{minipage}[t]{0.005\linewidth}
    ~
\end{minipage}
\begin{minipage}[t]{0.49\linewidth}
    \centering
\includegraphics[width=0.8\textwidth,trim=20 25 25 15]{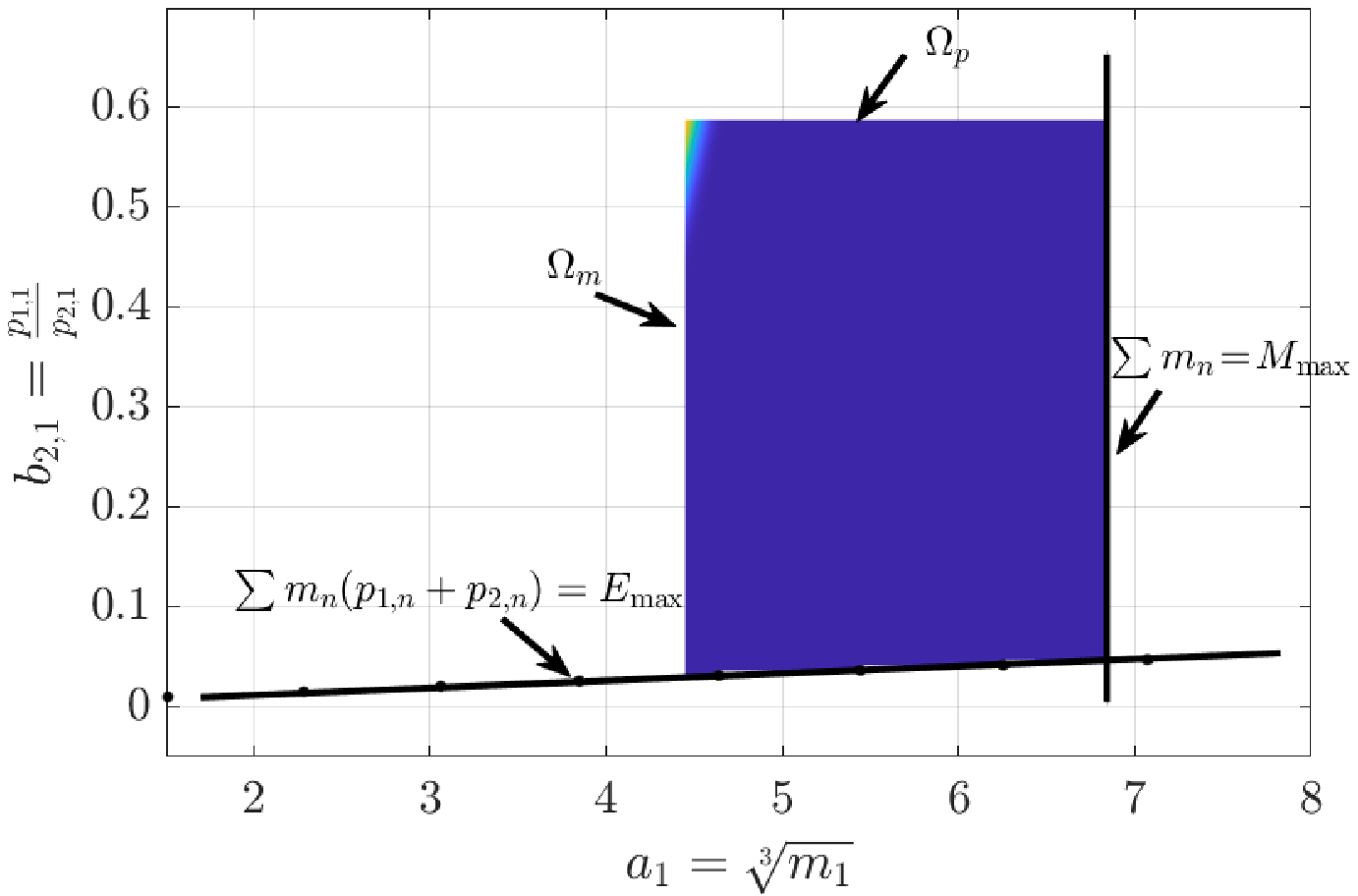}
\caption{Feasible set of~\eqref{problem:err_tau_ab} with 4 users and channel realization $\mathbf{z}=\{0.5,1.5;0.5,1.5\}$,  where $a_1=\sqrt[3]{m_1}$ and $b_{2,1}=\frac{p_{1,1}}{p_{2,1}}$ are presented as variables.}
\label{fig:feasible_ab}
\end{minipage}
\vspace{-20pt}
\end{figure}

\vspace{-17pt}
\subsection{Results Validation per Channel Realization}
In both reliability- and effective capacity-oriented designs, we introduce new variables $a_n=\sqrt[3]{m_n}$, $b_{1,n}=1/p_{1,n}$ and $b_{2,n}=p_{1,n}/p_{2,n}$ to tackle the issue of non-convex feasible set in \eqref{problem:err_tau} and~\eqref{problem:sub}. Therefore, we consider a scenario with 4 users as an example to show the feasible set with $m_1$ and $p_{1,1}$ in Fig.~\ref{fig:feasible_mp}, as well as its corresponding feasible set with $a_1$ and $b_{2,1}$ in Fig.~\ref{fig:feasible_ab}. The channel realizations of four users are set as $\mathbf{z}=\{0.5,1.5,0.5,1.5\}$. In order to reduce the dimensions of variables, we let $p_{2,n}=P_{\max}-p_{1,n}$ and $m_2=m_1$. Based on these setups, we can observe that the non-convexity arises from the energy budget constraint. The boundary of this constraint is actually the positive part of the rectangular hyperbola in the form of $f(x,y)=x/y$. Moreover, $\mathbf{\Omega}_p$ and $\mathbf{\Omega}_m$ construct a threshold for power and blocklength according to~\eqref{eq:omega_m} and~\eqref{eq:omega_p}, which have a linear boundary. After replacing the variables with new variables, the constraint becomes convex while the convexity of the rest constraint remains is unaffected. As a result, the feasible set becomes convex, confirming Lemma~\ref{lemma:err_convex}. It should be pointed out that similar shapes for the feasible set with different numerical values can be observed in Problem~\eqref{problem:sub} and Problem~\eqref{problem:sub_ab} in the effective capacity-oriented design, if $\lambda_{\rm E}$ is given in the same system setups. 

{Next, we move on to the power allocation strategies for NOMA scheme, as well as the blocklength allocation strategies for TMDA scheme for the two-user case.   
In Fig.~\ref{fig:err_inst_h}, we plot the optimal maximal error probability for both NOMA and TDMA scheme with instantaneous channel gains. Moreover, we set the fixed channel gain $z_{2,1}=2$ and vary the channel gain $z_{1,1}$, while letting the maximal transmit power to be $P_{\max}=\{0.8, 1, 1.2\}$W. Although the error probability of both schemes improves with increase of the channel gain, their performance behaviors are quite different. In particular, we can clearly observe the trade-off between TDMA and NOMA scheme. NOMA scheme benefits from the diversity of the channel gain among the users. However, when the two channel gains are even, NOMA scheme suffers from the interference. The gain from utilizing the blocklength can no longer compromise the negative influence of error propagation from the imperfect SIC. In such case, TDMA scheme becomes a better choice since there is no interference at all. 
{ This observation further motivates us to investigate the hybrid NOMA-TDMA scheme with the reliability-oriented design framework for the scenarios with sporadic traffics since it combines both the advantage of NOMA and TDMA scheme.  }}
\begin{figure}[!t]
    \centering
\includegraphics[width=0.4\textwidth,trim=0 10 0 15]{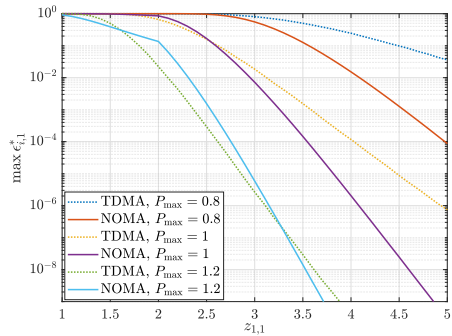}
\caption{Optimal maximal error probability for both TDMA and NOMA schemes against instantaneous channel gain $z_{1,1}$ under maximal transmit power constraint $P_{\max}=\{0.8, 1, 1.2\}$W, while $z_{2,1}=2$ is fixed.}
\label{fig:err_inst_h}\vspace*{-10pt}
\end{figure}

\begin{figure}[!t]
\begin{minipage}[t]{0.49\linewidth}
    \centering
\includegraphics[width=0.8\textwidth,trim=15 20 30 20]{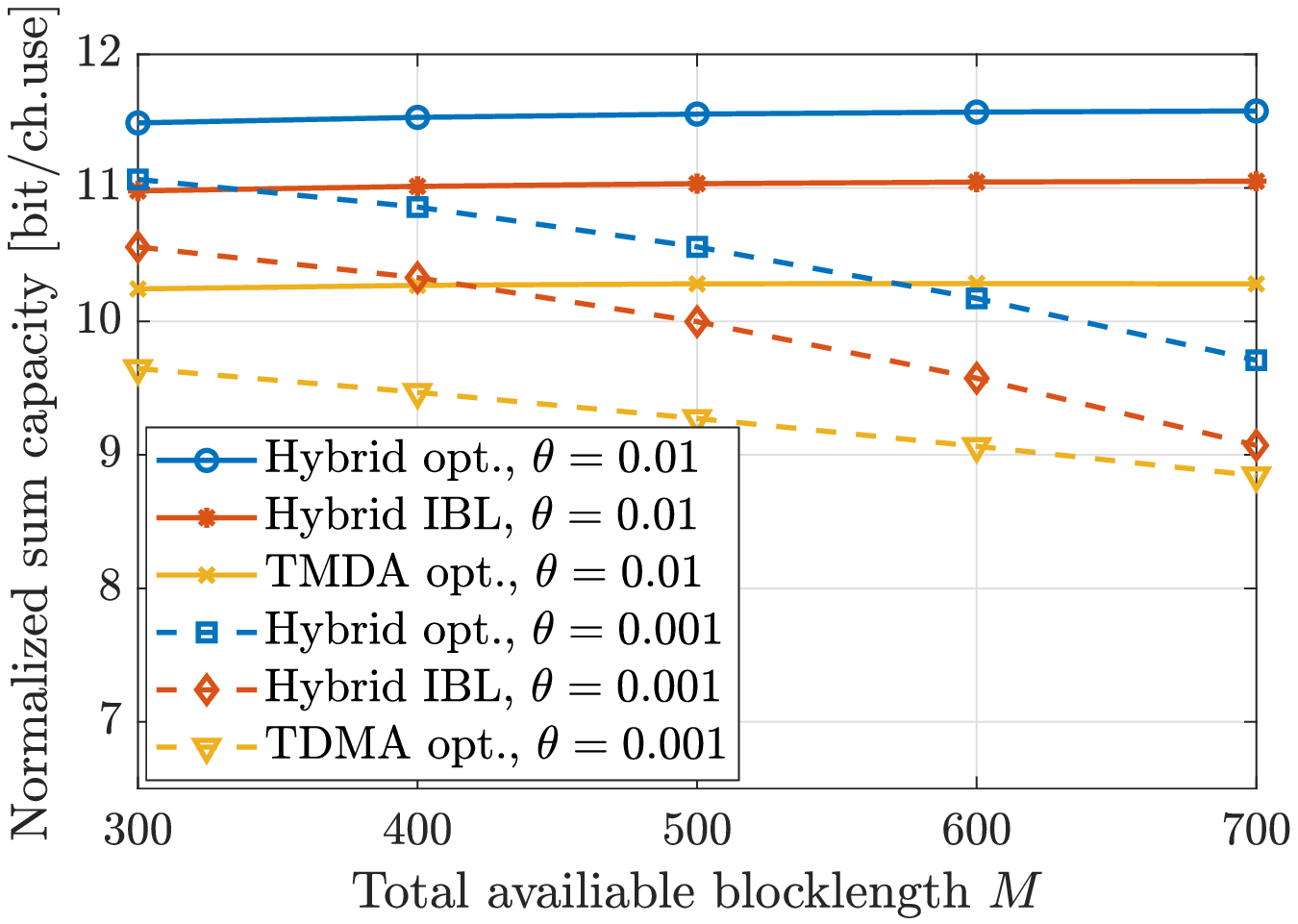}
\caption{Normalized sum capacity against target error probability $\bar{\varepsilon}$ with 4 user pairs.}
\label{fig:ec_vs_M}
\end{minipage}
\begin{minipage}{0.005\linewidth}
~
\end{minipage}
\begin{minipage}[t]{0.49\linewidth}
    \centering
\includegraphics[width=0.8\textwidth,trim=15 20 30 20]{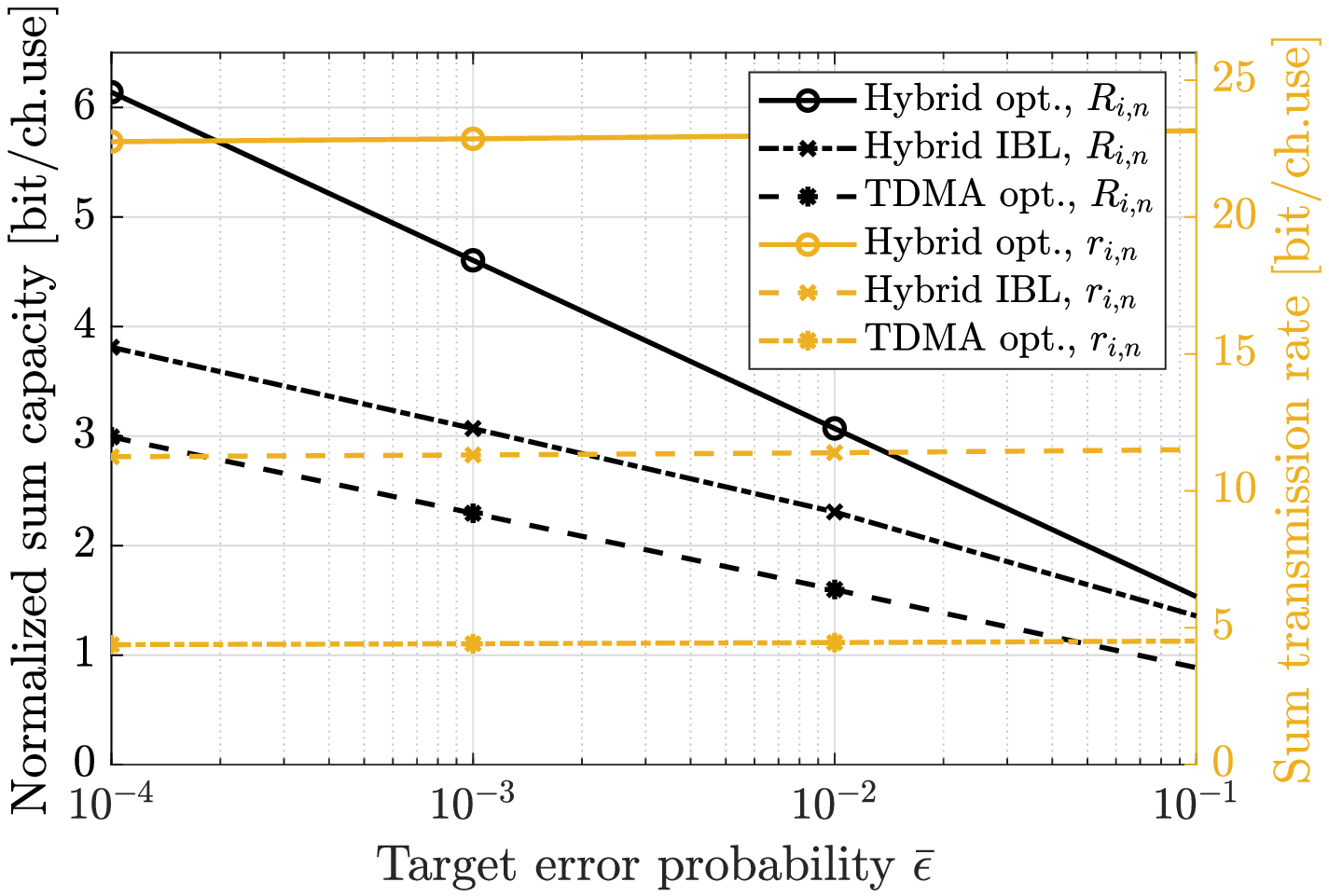}
\caption{Normalized sum capacity and sum transmission rate against total available blocklength with 4 user pairs. QoS component is set as $\theta=\{10^{-2},10^{-3}\}$.}
\label{fig:err_comparison}
\end{minipage}
\vspace{-29pt}
\end{figure}

\vspace{-10pt}
\subsection{Impact of Parameters on Effective Capacity}
In previous subsections, we have discussed the validation of our proposed approaches and resource allocation strategies. We can draw  similar conclusions for the effective capacity-oriented design, e.g., the accuracy of SINR approximation in Fig.~\ref{fig:err_comparison_adaptive}, the feasible set issue in Fig.~\ref{fig:feasible_mp}. %
In addition to these discussion, we should also investigate the impact of parameters that are unique in the effective capacity-oriented design.

In Fig.~\ref{fig:ec_vs_M}, we plot the normalized sum effective capacity against total available blocklength $M$ with 4 user pairs while varying the QoS component $\theta=\{0.01,0.001\}$. We also show the performance of two benchmarks to demonstrate the advantage of our proposed design of the hybrid NOMA-TDMA scheme. 
Unlike the reliability-oriented design, increasing $M$ does not always improve the effective capacity. Instead, the improvement depends on the value of $\theta$. For example, for $\theta=0.001$, the effective capacity actually diminishes, since the transmission rate can hardly be improved with a larger number of $M$ while the queue delay is proportional to $M$. 
In fact, if $\theta$ is sufficiently high, e.g., $\theta\approx 1$, there is almost no QoS requirement. Then, normalized effective capacity will degenerate into a modified transmission rate as $\frac{1}{M}\sum R_{i,n}\approx \sum r_{i,n}$. Furthermore, as expected, our proposed design outperforms both benchmarks regardless of the considered setting. However, the differences are not as dramatic as the differences in the reliability-oriented design. This is due to the fact that the objective is sum of effective capacities $\sum R_{i,n}$ instead of the maximal function of error probability. 
{ Therefore, compared with the conventional multi-access schemes, the advantage of the proposed approach in the queue-influenced scenarios is still preserved. }

To investigate the impact of the target error probability $\bar{\varepsilon}$ on the system and demonstrate the performance difference between normalized effective capacity and coding rate under the considered hybrid NOMA-TDMA scheme, we plot the normalized sum capacity and sum transmission rate versus target error probability $\bar{\varepsilon}$. As  $\bar{\varepsilon}$ increases, the FBL transmission rate always increases for all users according to~\eqref{eq:coding_rate_1}. However, this is not true for the effective capacity. $R_{i,n}$ is heavily influenced  by $\bar{\varepsilon}$, since higher $\bar{\varepsilon}$ means that the queue delay requirement is more likely to be violated. In the extreme case with $\bar{\varepsilon}=1$, the FBL transmission rate simply coincides with the Shannon capacity, i.e., $r_{i,n}=\log_2(1+\gamma_{i,n})$, while effective capacity is zero, since all transmissions will violate the QoS constraints. Indeed, this has motivated us to investigate the performance of the effective capacity-oriented design in this work.

\vspace{-10pt}
\section{Conclusion}
Hybrid NOMA-TDMA scheme is one of the most promising approaches  to   provide massive  connectivity  to devices in   future IoT networks. 
{In this work, we studied the hybrid NOMA-TDMA scheme in the  uplink communications, where transmissions are performed with FBL codes due to the low-latency requirements.}  
We provided    design frameworks for two types of latency-sensitive applications: First, for latency-critical applications where the main concern is the reliability of the one-shot transmission, we proposed a reliability-oriented design from a pure physical-layer perspective. Subsequently, for the other type of latency-sensitive  applications in which the random queuing behavior plays a role, we introduce a design framework   maximizing the link-layer performance, namely effective capacity. 
In particular,  for the reliability-oriented design, we aim at minimizing the maximal error probability among the users by jointly allocating the blocklength and transmit power of each user. {We decomposed the original problem into several solvable sub-problems and proved the convexity of those problems.} 
For the effective capacity-oriented design, we  aimed at maximizing the sum of effective capacities of all users. After  leveraging the Lagrange dual method, %
we characterized the joint concavity of the transmission rate, as well as the effective capacity and showed that the
strong duality holds, and the problem can be efficiently solved.  
Via simulations, we validated our analytical models and demonstrated the advantages of our proposed approaches in comparison to   benchmarks. 
We also revealed the impact of various parameters on the system performance and its influences on the resource allocation strategies.   Especially, a significant performance gap is observed between IBL-based and FBL-based designs, which confirms the necessity of  taking the FBL impact into account in the design frameworks.

We  conclude the   paper by   reiterating that the proposed analytical model has a high extensibility.
{First, although we studied the uplink NOMA scheme, the optimal power control strategy for NOMA scheme can be also be extended to a downlink NOMA scheme or with even higher power levels.}  
More importantly, the joint convexity shown in this work has a high potential to facilitate the designs with  a similar problem structure, i.e.,    joint power-blocklength allocation where the transmit power contributes to both signal and interference. 
For instance, this can be directly applied to the joint power-blocklength allocation in NOMA-relaying   and adaptive NOMA/OMA schemes. 
{ Moreover, our design framework can also be extended to a random access scheme, where the time slots are no longer pre-allocated, i.e., it becomes a hybrid NOMA-ALOHA scheme. }

\vspace{-10pt}
\appendices
\vspace{-5pt}
\section{Proof of Lemma~\ref{lemma:err_convex}}\vspace{-5pt}
\begin{proof}
We first investigate the joint convexity of the objective function. Since the objective is a max function, the convexity can be shown by proving the convexity of error probability for each user $\varepsilon_{i,n}$, i.e., the max function %
is convex if all the components %
are convex. 
First, consider $\mathbbm{1}_{\varepsilon<1}(\varepsilon_{1,n})=1$. Recall that the error probability for user 1 in any pair $n$ is given by %
    $\varepsilon_{1,n}=\varepsilon_{2|1,n}+\varepsilon_{1|1,n}$.
Then, we can further investigate the joint convexity of both $\varepsilon_{2|1,n}$ and $\varepsilon_{1|1,n}$. In particular, for the Hessian matrix $\mathbf{H}(\varepsilon_{2|1,n})$, we have:\vspace{-10pt}
\begin{equation}
    \mathbf{H}( \varepsilon_{2|1,n}(\mathbf{a},\mathbf{b}))=\mathbf{H}(\nabla^2\varepsilon_{2|1,n}(a_n,b_{1,n})).\vspace{-5pt}
\end{equation}\vspace{-5pt}
The equality holds since $\varepsilon_{2|1,n}$ only depends on the blocklength $m_n=a^3_n$ and the SINR $\gamma_{2|1}=\frac{z_{2,n}}{z_{1,n}b_{2,n}}$. Then, the first order and second order derivatives of $\varepsilon_{2|1,n}$ w.r.t. $m_n$ are given by:\vspace{-5pt}
\begin{equation}
    \frac{\partial \varepsilon_{2|1,n}}{\partial m_n}= -\frac{\ln 2}{\sqrt{2\pi m_n V_{2|1,n}}}e^{-\omega_{2|1,n}}\left(
        \mathcal{C}_{2|1,n}+m_n
     \right)\leq 0,\vspace{-5pt}
\end{equation}\vspace{-5pt}
\begin{equation}
    \frac{\partial^2 \varepsilon_{2|1,n}}{\partial m^2_n}= \omega_{2|1,n}\frac{\ln 2}{\sqrt{8\pi m_n V_{2|1,n}}}e^{-\omega_{2|1,n}}\left(
        \frac{\mathcal{C}_{2|1,n}-3\frac{d}{m_n}}{m_n}
     \right)\geq 0,\vspace{-5pt}
\end{equation}\vspace{-5pt}
where we define $V_{2|1,n}=V(\gamma_{2|1,n})$, $\mathcal{C}_{2|1,n}=\mathcal{C}({2|1,n})$ and $\omega_{2|1,n}=\sqrt{\frac{m_n}{V_{2|1,n}}} \left(\mathcal{C}_{2|1,n}-\frac{d}{m_n}\right)$ to simplify the notation. 
Furthermore, the second-order partial derivatives w.r.t. $a_n$ are expressed as:\vspace{-5pt}
\begin{equation}
\begin{split}
    \frac{\partial^2 \varepsilon_{2|1,n}}{\partial a^2_n}&=\frac{\partial \varepsilon_{2|1,n}}{\partial m_n}\frac{\partial^2 m_n}{\partial a^2_n}
    +\frac{\partial^2 \varepsilon_{2|1,n}}{\partial m^2_n}\left(
    \frac{\partial m_n}{\partial a_n}
    \right)^2\\
    &
    = \frac{3a_ne^{-\wren}}{\sqrt{2\pi}}
    \Bigg(\underbrace{\frac{\partial \wren}{\partial m_n}}_{\geq 0}
        \left(
            3m_n\wren \frac{\partial \wren}{\partial m_n}-2
        \right)
        - \underbrace{3m_n\frac{\partial^2 \wren}{\partial m^2_n}}_{\leq 0}
    \Bigg)\\
    &\geq \frac{3a_ne^{-\wren}}{\sqrt{2\pi}}\frac{\partial \wren}{\partial m_n}\!
    \Bigg(\!
        3\wren m_n(\gamma_{2|1,n}\!+\!1)\!\sqrt{\frac{m_n}{\gamma_{2|1,n}(\gamma_{2|1,n}\!+\!2)}}\Cren\!-\!2\!+\!3d\wren\sqrt{\frac{m_n}{\Vren}}
    \Bigg)\\
     \overset{\gamma_{2|1,n}\geq 1}&{\geq}\frac{3a_ne^{-\wren}}{\sqrt{2\pi}}\frac{\partial \wren}{\partial m_n}
     \left(
         3\cdot1.25\cdot 1(1+1)\sqrt{1/3}\log_22-2
     \right)\geq 0. \nonumber
\end{split}
\end{equation}
Therefore, $\varepsilon_{2|1,n}$ is convex in $a_{n}$, i.e., in $\mathbf{a}$. A similar conclusion can be drawn for $\mathbf{b}$ by showing:\vspace{-10pt}
\begin{equation}
\begin{split}
\!\!\!\!\!\!\frac{\partial{\wren}}{\partial{\gamma_{2|1,n}}}
&=-\frac{e^{-\wren}}{\sqrt{2\pi}}\!
\left(\!
    \frac{\sqrt{m_n/\Vren}\left(\gamma_{2|1,n}^2\!+\!2\gamma_{2|1,n}\!-\!\ln(\gamma_{2|1,n}\!+\!1)\right)}{(\gamma_{2|1,n}^2\!+\!2\gamma_{2|1,n})(\gamma_{2|1,n}\!+\!1)}
    \!+\!\frac{d\ln2}{\!\sqrt{m_n\Vren^3}}\frac{1}{(1\!+\!\gamma_{2|1,n})^3}\!
\right)\!\!\!\!\\
\overset{\gamma_{2|1,n}\geq 1}&{\leq}0,
\end{split}
\end{equation}\vspace{-5pt}
\begin{equation}
\begin{split}
\!\!\frac{\partial^2{\wren}}{\partial{\gamma^2_{2|1,n}}}
&=\frac{\wren}{\sqrt{2\pi}}e^{-\wren}
\left(\!
    \frac{\sqrt{m_n}\left(-(\gamma_{2|1,n}\!+\!1)^3\!+\!\frac{1}{\gamma_{2|1,n}\!+\!1}\right. %
\left.\!+\!3\ln2(\gamma_{2|1,n}\!+\!1)\Big(\Cren\!-\!\frac{d}{m_n}\Big)\right)}{(\gamma_{2|1,n}(\gamma_{2|1,n}+2))^{\frac{5}{2}}}\!
\right)\\
\overset{\gamma_{2|1,n}\geq 1}&{\geq}0,
\end{split}\vspace{-10pt}
\end{equation}
as well as the second-order partial derivative w.r.t. $b_{2,n}$:\vspace{-5pt}
\begin{equation}
\begin{split}
\frac{\partial^2 \varepsilon_{2|1,n}}{\partial b^2_{2,n}}&=\frac{\partial \varepsilon_{2|1,n}}{\partial \gamma_{2|1,n}}\frac{\partial^2 \gamma_{2|1,n}}{\partial b^2_{1,n}}
+
\frac{\partial^2 \varepsilon_{2|1,n}}{\partial \gamma^2_{2|1,n}}
    \left(
        \frac{\partial \gamma_{2|1,m}}{\partial b_{2,n}}
    \right)^2\\
    &=\frac{\gamma_{2|1,n}}{b^2_{2,n}}\frac{e^{-\wren}}{\sqrt{2\pi}}
    \Big(
        \underbrace{\frac{\partial \wren}{\partial \gamma_{2|1,n}}}_{\geq 0}
        \left(
            \wren\gamma_{2|1,n}
        \frac{\partial \wren}{\partial \gamma_{2|1,n}}-1
        \right)-\gamma_{2|1,n}
        \underbrace{\frac{\partial^2 \wren}{\partial \gamma^2_{2|1,n}}}_{\leq 0}
    \Big)\\
    &\geq \frac{\gamma_{2|1,n}}{b^2_{2,n}}\frac{e^{-\wren}}{\sqrt{2\pi}}
    \frac{\partial \wren}{\partial \gamma_{2|1,n}}
    \left(
        1.25\cdot 1\cdot\frac{4}{3}\frac{\sqrt{5}}{(1\!+\!2)(1\!+\!1)}(1\!+\!2\!-\!\log_2(1\!+\!1))\!-\!1
    \right)\geq 0,
\end{split}
\end{equation}
where the inequality holds for any reasonable allocated blocklength $m_n\geq 5$. Therefore, $\varepsilon_{2|1,n}$ is also convex in $b_{2,n}$, i.e., in $\mathbf{b}$ (note that $\varepsilon_{2|1,n}$ is independent to $b_{1,n}$). As a result, we can leverage~\cite[Lemma~1]{joint_letter} to show that the Hessian matrix $\mathbf{H}(\varepsilon_{2|1,n})$ is positive semi-definite, since we have proven $\frac{\partial^2 \varepsilon_{2|1,n}}{\partial b^2_{2,n}}\geq 0$ and $\frac{\partial^2 \varepsilon_{2|1,n}}{\partial a^2_{n}}\geq 0$. In other words, $\varepsilon_{2|1,n}$ is jointly convex in $\mathbf{a}$ and $\mathbf{b}$.

The joint convexity of $\varepsilon_{2|2,n}$  can be also shown using  exactly the same methodology by replacing $b_{2,n}$ with $b_{1,n}$ and $\gamma_{2|2,n}$ with $\gamma_{2|1,n}$ in the above derivations. We omit the details to avoid duplicating the proof. Moreover, it is trivial to show that for $\mathbbm{1}_{\varepsilon<1}(\varepsilon_{1,n})=0$, the error probability is constant, i.e., jointly convex in any variable. Consequently, as a sum of convex functions, {$\varepsilon_{1,n}=\varepsilon_{2|1,n}+\varepsilon_{2|1,n}$ is} also jointly convex in $\mathbf{a}$ and $\mathbf{b}$. In the meantime, the same conclusion can be drawn for $\varepsilon_{2,n}=\varepsilon_{2|1,n}$. Therefore, the objective function is jointly convex
regardless of the value of $\mathbbm{1}_{\varepsilon<1}(\varepsilon_{1,n})$. 
Furthermore, we can show that all constraints are also jointly convex, i.e., $\det \mathbf{H}(a^3_n)=6a_n\geq0$, $\det \mathbf{H}(a^3_n(\frac{1}{b_{1,n}b_{2,n}}+\frac{1}{b_{1,n}}))=\frac{3a^4_n}{b^4_{1,n}}\geq 0$, as well as $\det \mathbf{H}(\frac{1}{b_{1,n}b_{2,n}}+\frac{1}{b_{1,n}}=\frac{4b_{2,n}+3}{b^4_{1,n}b^4_{2,n}})\geq0$.

In summary, both the objective function and constraints are jointly convex in $\mathbf{b}$ and $\mathbf{a}$ within the feasible set. Thus, Problem~\eqref{problem:err_tau_ab} is a convex problem.
\end{proof}
\vspace{-10pt}
\section{Proof of Lemma~\ref{lemma:time_sharing}}\vspace{-10pt}
\begin{proof}
{Let $\mathbf{M}^*_{\text{x}}$ and $\mathbf{P}^*_{\text{x}}$ to be the optimal solutions to the Problem~\eqref{problem:effective_capacity_original} with a non-negative average energy budget $E^*_{\text{x}}$ while $\mathbf{M}^*_{\text{y}}$ and $\mathbf{P}^*_{\text{y}}$ to be the optimal solutions to~\eqref{problem:effective_capacity_original} with $E^*_{\text{y}}$. Let $R^*_x=\frac{1}{M}\sum^N_{n=1}\sum^2_{i=1}R_{i,n}(\mathbf{M}^*_{\text{x}},\mathbf{P}^*_{\text{x}})$ be the respective optimal values. Note that the resources are allocated frame-wise. Therefore, in any frame $\tau$, the optimal blocklength allocation $\mathbf{m}^*[\tau]$ and power allocation $\mathbf{p}^*[\tau]$ are constant and it holds $\frac{1}{L}\sum^L_{\tau=1}\mathbf{m}^*[\tau]\mathbf{p}^*[\tau]\leq E_{\max}$ over all frames.

Now, let $0\leq\nu\leq1$. Since it holds $L\to\infty$, we can take $\nu L$ frames, where the resource allocation corresponding to $\mathbf{M}^*_{\text{x}}$ and $\mathbf{P}^*_{\text{x}}$, with which it achieves the effective capacity of $\nu R^*_{\text{x}}$. Similarly, we take $(1-\nu)L$ frames and allocating the resources as $\mathbf{M}^*_{\text{y}}$ and $\mathbf{P}^*_{\text{y}}$. It can achieve $(1-\nu)R^*_{\text{y}}$. Then , the effective capacity over $L$ frames is $\nu R^*_{\text{x}}+(1-\nu)R^*_{\text{y}}$, while the average energy budget is lesser than $E^*_{\text{x}}+E^*_{\text{y}}$. As a result, the time-sharing property holds.}
\vspace{-10pt}
\end{proof}
\section{Proof of Lemma~\ref{lemma:r_convex}}\vspace{-10pt}
\begin{proof}
The first- and second-order derivatives of $r_{i,n}$ with respect to $m_{n}$ can be written as:
\begin{equation}
\label{eq:dr_dm}
\frac{\partial{r_{i,n}}}{\partial{m_n}}=\frac{1}{2}V_{i,n}^{\frac{1}{2}}m_n^{-\frac{3}{2}}Q^{-1}(\bar{\varepsilon}_{i,n})\log_2e\geq 0,
\vspace{-5pt}
\end{equation}
\begin{equation}
\label{eq:d2r_dm2}
\frac{\partial^2{r_{i,n}}}{\partial{m_n}^2}=-\frac{3}{4}V_{i,n}^{\frac{1}{2}}m_n^{-\frac{5}{2}}Q^{-1}(\bar{\varepsilon}_{i,n})\log_2e\leq 0.
\vspace{-10pt}
\end{equation}
Note that $\frac{\partial{r_{i,n}}}{\partial {m_{\tilde{n}} }}=0$, if $n\neq \tilde{n}$. Therefore, these characterizations imply that $r_{i,n}$ is non-decreasing and concave in $\mathbf{m}$. Next, to simplify the notation, we define $\gamma_{1,n}=\gamma_{1|1,n}$ and $\gamma_{2,n}=\gamma_{2|1,n}$, respectively. 
Likewise, we prove that $-r_{i,n}$ is also non-decreasing and concave in $\boldsymbol{\gamma}$ by showing:\vspace{-5pt}
\begin{equation}
\begin{split}
\label{eq:dr_dgamma}
\frac{\partial{r_{i,n}}}{\partial{\gamma_{i,n}}}&=\!\frac{\partial{\mathcal{C}_{i,n}
}}{\partial{\gamma_{i,n}}}\!-\!\frac{1}{2\ln 2}\sqrt{\frac{1}{V_{i,n}m}}
Q^{-1}(\bar{\varepsilon}_{i,n})\frac{\partial{V_{i,n}}}{\partial{\gamma_{i,n}}}
\!=\!\frac{1}{(\gamma_{i,n}+1)\ln2}\big(1\!-\!\frac{m_n^{-\frac{1}{2}}Q^{-1}(\bar{\varepsilon}_{i,n})}{\sqrt{\gamma_{i,n}^2+2\gamma_{i,n}}(\gamma_{i,n}+1)}\big)\\
\overset{\gamma_{i,n}\geq0}&{\geq}\frac{1}{(\gamma_{i,n}+1)\ln2}{\big(1-\frac{1}{2\sqrt{3}}\big)}\geq0, \nonumber
\end{split}
\end{equation}\vspace{-5pt}
\begin{equation}
\begin{split}
\label{eq:d2r_dgamma2}
\frac{\partial^2{r_{i,n}}}{\partial{\gamma_{i,n}}^2}&=\frac{1}{(\gamma_{i,n}+1)^2\ln2}\bigg(-1+\frac{(\gamma_{i,n}+1)m_n^{-\frac{1}{2}}Q^{-1}(\bar{\varepsilon}_{i,n})}{(\gamma_{i,n}^2+2\gamma_{i,n})^{\frac{3}{2}}}+\frac{2m_n^{-\frac{1}{2}}Q^{-1}(\bar{\varepsilon}_{i,n})}{(\gamma_{i,n}^2+2\gamma_{i,n})^{\frac{1}{2}}(\gamma_{i,n}+1)}\bigg)\\
\overset{\gamma_{i,n}\geq0}&{\leq} \frac{1}{(\gamma_{i,n}+1)^2\ln2}\bigg(-1+\frac{2m_n^{-\frac{1}{2}}Q^{-1}(\bar{\varepsilon}_{i,n})}{3\sqrt{3}}+\frac{2m_n^{-\frac{1}{2}}Q^{-1}(\bar{\varepsilon}_{i,n})}{2\sqrt{3}}\bigg)\\
&{\leq}\frac{1}{(\gamma_{i,n}+1)^2\ln2}{\bigg(-1+\frac{5}{3\sqrt{3}}\bigg)}\leq0.
\end{split}\vspace{-10pt}
\end{equation}
To fully characterize the joint convexity, we also need the partial derivative $\frac{\partial^2{(-r_{i,n})}}{\partial{m_n}\partial{\gamma_{i,n}}}$, i.e.,\vspace{-5pt}
\begin{equation}
    \label{eq:d2r_dmdgamma}
    \frac{\partial^2{(-r_{i,n})}}{\partial{m_n}\partial{\gamma_{i,n}}}=\frac{1}{2(\gamma_{i,n}+1)^3\ln2}v_{i,n}^{-\frac{1}{2}}m_n^{-\frac{3}{2}}Q^{-1}(\bar{\varepsilon}_{i,n})\geq 0.\vspace{-5pt}
\end{equation}

Therefore, combing~\eqref{eq:d2r_dm2},~\eqref{eq:d2r_dgamma2} and~\eqref{eq:d2r_dmdgamma}, the determinant  of Hessian matrix of $-r_{i,n}$ w.r.t. $\mathbf{m}$ and $\boldsymbol{\gamma}$ is given by:\vspace{-5pt}
\begin{equation}
\begin{split}
\det \mathbf{H_{\mathbf{m},\boldsymbol{\gamma}}}(-r_{i,n})&=\det \left( %
\begin{array}{cc} %
\frac{\partial^2{(-r_{i,n})}}{\partial{m_n^2}}& \frac{\partial^2{(-r_{i,n})}}{\partial{m_n}\partial{\gamma_{i,n}}} \\ %
\frac{\partial^2{(-r_{i,n})}}{\partial{\gamma_{i,n}}\partial{m_n}}&\frac{\partial^2{(-r_{i,n})}}{\partial{\gamma^2_{i,n}}}\\ %
\end{array}
\right)%
=\frac{\partial^2{(-r_{i,n})}}{\partial{m_n^2}}\frac{\partial^2{(-r_{i,n})}}{\partial{\gamma_{i,n}}^2}-\frac{\partial^2{(-r_{i,n})}}{\partial{m_n}\partial{\gamma_{i,n}}}\frac{\partial^2{(-r_{i,n})}}{\partial{\gamma_{i,n}}\partial{m_n}}\\
&=\frac{m_n^{-\frac{5}{2}}Q^{-1}(\bar{\varepsilon}_{i,n})}{( \gamma_{i,n}+1)^3 (\ln2)^2}\bigg(\frac{3\sqrt{ \gamma_{i,n}^2+2 \gamma_{i,n}}}{4}-\frac{3m_n^{-\frac{1}{2}}Q^{-1}(\bar{\varepsilon}_{i,n})( \gamma_{i,n}+1)}{4( \gamma_{i,n}^2+2 \gamma_{i,n})}\\
&\qquad \qquad \qquad \qquad -\frac{3m_n^{-\frac{1}{2}}Q^{-1}(\bar{\varepsilon}_{i,n})}{2( \gamma_{i,n}+1)}-\frac{m_n^{-\frac{1}{2}}Q^{-1}(\varepsilon_i)}{4( \gamma_{i,n}+1)( \gamma_{i,n}^2+2 \gamma_{i,n})}\bigg)\\
&\geq\frac{m_n^{-\frac{5}{2}}Q^{-1}(\bar{\varepsilon}_{i,n}) }{(\gamma_{i,n}+1)^3 (\ln2)^2}\bigg(\frac{3\sqrt{3}}{4}-\frac{31m_n^{-\frac{1}{2}}Q^{-1}(\bar{\varepsilon}_{i,n})}{24}\bigg)\geq 0.%
\end{split}\vspace{-5pt}
\end{equation}
Thus, $-r_{i,n}$ is jointly convex in $\mathbf{m}$ and $\boldsymbol{\gamma}$.
\end{proof}\vspace{-10pt}
\section{Proof of Lemma~\ref{lemma:EC_convex}}
\begin{proof}
First, we introduce the auxiliary function $K_{i,n}=m_nr_{i,n}$, in the following:\vspace{-5pt}
\begin{equation}
\begin{split}
\frac{\partial{(-R_{i,n})}}{\partial K_{i,n}}
=&\frac{1}{\theta_{i,n}}\frac{\frac{\partial{\mathbb{E}\big[e^{-\theta_{i,n} m_nr_{i,n}}(1-\varepsilon_i)+\varepsilon_i\big]}}{\partial{K_{i,n}}}}{\mathbb{E}\big[e^{-\theta_{i,n}m_n r_{i,n}}(1-\varepsilon_i)+\varepsilon_i\big]}=-\frac{e^{-\theta_{i,n}m_nr_{i,n}}(1-\varepsilon_i)f_{\mathbf{Z}}(\mathbf{z}[\tau])}{\mathbb{E}\big[e^{-\theta_{i,n}m_nr_{i,n}}(1-\varepsilon_i)+\varepsilon_i\big]}\leq 0,
\end{split}
\end{equation}
\begin{equation}
\label{eq:d2R_dm2}
\begin{split}
\frac{\partial^2{(-R_{i,n})}}{\partial K_{i,n}^2}
=&\frac{\theta_{i,n}e^{-\theta_{i,n}m_nr_{i,n}}(1-\varepsilon_i)f_{\mathbf{Z}}(\mathbf{z}[\tau])}{\mathbb{E}^2\big[e^{-\theta_{i,n}m_nr_{i,n}}(1-\varepsilon_i)+\varepsilon_i\big]}\\
&\cdot\Big(\mathbb{E}\big[e^{-\theta_{i,n}m_nr_{i,n}}(1-\varepsilon_i)+\varepsilon_i\big]-e^{-\theta_{i,n}m_nr_{i,n}}(1-\varepsilon_i)f_{\mathbf{Z}}(\mathbf{z}[\tau])\Big)\\
&\geq \frac{\theta_{i,n}e^{-\theta_{i,n}m_nr_{i,n}}(1-\varepsilon_i)f_{\mathbf{Z}}(\mathbf{z}[\tau])}{\mathbb{E}^2\big[e^{-\theta_{i,n}m_nr_{i,n}}(1-\varepsilon_i)+\varepsilon_i\big]}\varepsilon_if_{\mathbf{Z}}(\mathbf{z}[\tau])\geq 0.
\end{split}\vspace{-5pt}
\end{equation}\vspace{-5pt}
Above inequality holds, since $K_{i,n}$ only depends on the current channel realization $\mathbf{z}[\tau]$. Then, we can show that:
\begin{equation}
\label{eq:d2R_dgamma2}
\begin{split}
\frac{\partial^2{(-R_{i,n})}}{\partial{m_n^2}}&=
\frac{\partial^2(-R_{i,n})}{\partial K_{i,n}^2}\left(\frac{\partial K_{i,n}}{\partial m_n}\right)^2
+\frac{\partial(-R_{i,n})}{\partial K_{i,n}}\frac{\partial^2 K_{i,n}}{\partial m_n^2}\\
&=\frac{\partial^2(-R_{i,n})}{\partial K_{i,n}^2}\bigg(r_{i,n}^2+2r_{i,n}m_n+m^2\frac{\partial r_{i,n}}{\partial m}\bigg)+\frac{\partial(-R_{i,n})}{\partial K_{i,n}}
\bigg(
    2\frac{\partial r_{i,n}}{\partial m_n}+m\frac{\partial^2 r_{i,n}}{\partial m_n^2}
\bigg)\overset{\eqref{eq:dr_dm},~\eqref{eq:d2r_dm2}}{\geq} 0,\nonumber
\end{split}\vspace{-5pt}
\end{equation}\vspace{-5pt}
\begin{equation}
\frac{\partial^2{(-R_{i,n})}}{\partial{\gamma_{i,n}}^2}=m^2\frac{\partial^2(-R_{i,n})}{\partial K_{i,n}^2}
\bigg(
    \frac{\partial r_{i,n}}{\partial \gamma_{i,n}}
\bigg)^2
+m\frac{\partial(-R_{i,n})}{\partial K_{i,n}}\frac{\partial^2 r_{i,n}}{\partial \gamma_{i,n}^2}
\overset{\eqref{eq:dr_dgamma},~\eqref{eq:d2r_dgamma2}}{\geq} 0. \nonumber
\end{equation}
In other words, $-R_{i,n}$ is convex in both $m_n$ and $\gamma_{i,n}$, respectively. In order to further prove the joint convexity, we investigate the   determinant of the Hessian matrix with the help of the auxiliary function $K_{i,n}$: 
\vspace{-10pt}
\begin{equation}
    \label{eq:Hessian_R}
    \begin{split}
        &\det \mathbf{H_{\mathbf{m},\boldsymbol{\gamma}}}(-R_{i,n})%
        =\det \left( %
\begin{array}{cc} %
\frac{\partial^2{(-R_{i,n})}}{\partial{m_n^2}}& \frac{\partial^2{(-R_{i,n})}}{\partial{m_n}\partial{\gamma_{i,n}}} \\ %
\frac{\partial^2{(-R_{i,n})}}{\partial{\gamma_{i,n}}\partial{m_n}}&\frac{\partial^2{(-R_{i,n})}}{\partial{\gamma^2_{i,n}}}\\ %
\end{array}
\right)=\frac{\partial^2{(-R_{i,n})}}{\partial{m_n^2}}\frac{\partial^2{(-R_{i,n})}}{\partial{\gamma_{i,n}}^2}-\frac{\partial^2{(-R_{i,n})}}{\partial{m_n}\partial{\gamma_{i,n}}}\frac{\partial^2{(-R_{i,n})}}{\partial{\gamma_{i,n}}\partial{m_n}}\\
=&
\underbrace{\bigg(
    \frac{\partial^2{(-R_{i,n})}}{\partial K_{i,n}^2} \frac{\partial{K_{i,n}}}{\partial{m_n}}
    \frac{\partial{K_{i,n}}}{\partial{\gamma_{i,n}}}
\bigg)^2
    \!-\!\bigg(
    \frac{\partial^2{(-R_{i,n})}}{\partial K_{i,n}^2} \frac{\partial{K_{i,n}}}{\partial{m_n}}
    \frac{\partial{K_{i,n}}}{\partial{\gamma_{i,n}}}
\bigg)^2}_{=0}\\
&\quad +
\bigg(
    \frac{\partial{(-R_{i,n})}}{\partial K_{i,n}}
\bigg)^2
\underbrace{\bigg(
    \frac{\partial^2{K_{i,n}}}{\partial{m_n^2}}\frac{\partial^2{K_{i,n}}}{\partial{\gamma_{i,n}^2}}
    -\frac{\partial^2{K_{i,n}}}{\partial{m_n}\partial{\gamma_{i,n}}}
\bigg)}_{\geq 0}
-\frac{\partial{(-R_{i,n})}}{\partial K_{i,n}}
\frac{\partial^2{(-R_{i,n})}}{\partial K_{i,n}^2}\\
&\cdot
\Bigg(
    \frac{\partial{K_{i,n}}}{\partial{m_n}}
    \frac{\partial{K_{i,n}}}{\partial{\gamma_{i,n}}}
    \frac{\partial^2{K_{i,n}}}{\partial{\gamma_{i,n}}\partial{m_n}}
    -
    \bigg(
        \frac{\partial{K_{i,n}}}{\partial{\gamma_{i,n}}}
    \bigg)^2
    \frac{\partial^2{K_{i,n}}}{\partial{m_n^2}}
    -
    \bigg(
        \frac{\partial{K_{i,n}}}{\partial{m_n}}
    \bigg)^2
    \frac{\partial^2{K_{i,n}}}{\partial{\gamma_{i,n}^2}}
\Bigg)\\
\geq& \tilde{A}
\Bigg(
    \bigg(
        \frac{\partial{r_{i,n}}}{\partial{\gamma_{i,n}}}    
    \bigg)^2
    \bigg(
        2\underbrace{\frac{\partial{r_{i,n}}}{\partial{m_n}}}_{\geq 0}+2r_{i,n}
        +m_n \underbrace{\big(
            -m_n\frac{\partial^2{r_{i,n}}}{\partial{m_n^2}}
            -2\frac{\partial{r_{i,n}}}{\partial{m_n}}
        \big)}_{\geq 0}
    \bigg)\\
    +&\underbrace{\frac{\partial{r_{i,n}}}{\partial{\gamma_{i,n}}}}_{\geq 0}
    \underbrace{\frac{\partial^2{r_{i,n}}}{\partial{\gamma_{i,n}}\partial m_n}}_{\geq 0}
    \bigg(
        2r_{i,n}m_n+m_n^2\underbrace{\frac{\partial{r_{i,n}}}{\partial{m_n}}}_{\geq 0}
    \bigg)
    -\underbrace{\frac{\partial^2{r_{i,n}}}{\partial{\gamma_{i,n}^2}}}_{\leq 0}
    \bigg(
        r_{n,i}^2+
        (r_{n,i}+1)m_n\underbrace{\frac{\partial{r_{i,n}}}{\partial{m_n}}}_{\geq 0}
    \bigg)
\Bigg)\geq 0,
\end{split}
\end{equation}\vspace{-5pt}
where $\tilde{A}=-m\frac{\partial{(-R_{i,n})}}{\partial K_{i,n}}
\frac{\partial^2{(-R_{i,n})}}{\partial K_{i,n}^2}\geq 0$. Thus, according to~\eqref{eq:d2R_dm2} and~\eqref{eq:Hessian_R}, $-R_{i,n}$ is jointly convex in $m_n$ and $\gamma_{i,n}$.
\end{proof}
\bibliographystyle{IEEEtran}
\bibliography{reference_joint_NOMA}
\end{document}